\PassOptionsToPackage{tablesfirst}{endfloat}
\documentclass[useAMS,referee,usenatbib]{biom}
\usepackage{graphicx}
\usepackage{amsmath}
\usepackage{amssymb}
\usepackage{setspace}   
\usepackage[colorlinks=true,linkcolor=blue,citecolor=blue,urlcolor=blue]{hyperref}
\newcommand{\rev}[1]{#1}
\newcommand{\snew}[1]{#1}
\newcommand{\amk}[1]{#1}
\newtheorem{proposition}{Proposition}

\title[Objective-Driven Empirical Bayes Screening]{Multi-Objective Composite Longitudinal Biomarker Scores for Improved Cancer Risk Assessment}

\author{Bitan Sarkar$^{1}$, Ana Maria Kenney$^{2}$, James P. Long$^{1}$, Johannes F. Fahrmann$^{3}$, \\
\bf Samir Hanash$^{3}$, Kim-Anh Do$^{1}$, and Ehsan Irajizad$^{1,*}$\email{eirajizad@mdanderson.org} \\
$^{1}$Department of Biostatistics, The University of Texas MD Anderson Cancer Center, Houston, Texas, U.S.A. \\
$^{2}$Department of Statistics, University of California, Irvine, California, U.S.A. \\
$^{3}$Department of Clinical Cancer Prevention, The University of Texas MD Anderson \\
Cancer Center, Houston, Texas, U.S.A.}

\begin{document}

\renewcommand{\figureplace}{}
\renewcommand{\tableplace}{}

\pagerange{\pageref{firstpage}--\pageref{lastpage}}

\begin{abstract}
\amk{Repeated blood-based biomarker measurements can improve cancer risk assessment by capturing longitudinal changes missed by single-time-point analyses. Parametric Empirical Bayes (PEB) incorporates prior measurements to estimate individualized reference values, but existing implementations do not account for the time between measurements and rely on predefined panels with fixed combination rules. We developed improved Parametric Empirical Bayes (iPEB), which accounts for the intervals between serial measurements, adjusts for covariates, and performs feature selection and optimized biomarker combination. iPEB optimizes biomarker weights for specific clinical objectives, such as maximizing sensitivity at a prespecified specificity or diagnostic lead time. We evaluated iPEB through simulations and a real-world application using six protein biomarkers (pro-SFTPB, CEA, CA125, CYFRA~21-1, osteopontin, and HE4) from a case-control study nested within the Prostate, Lung, Colorectal, and Ovarian (PLCO) Cancer Screening Trial. The analysis included 324 lung cancer cases and 1{,}674 controls with at least two serial measurements; six centers were used for model development and four for independent validation. Optimized for sensitivity at 99\% specificity, iPEB achieved 24.2\% sensitivity in the independent test set, compared with 18.2\% for conventional PEB applied to the same four-marker panel. Optimized instead for lead time, iPEB added approximately 50 days of lead time at that stringent operating point. iPEB improved lung cancer risk assessment in independent PLCO data, supporting objective-driven optimization of longitudinal biomarkers for early detection.}
\end{abstract}


\maketitle

\section{Introduction}
\label{sec:intro}\label{firstpage}
Early cancer detection through accurate risk assessment is fundamental to population-level cancer screening. Effective screening strategies must identify preclinical disease sufficiently early to enable intervention while maintaining high specificity to minimize unnecessary follow-up procedures, patient anxiety, and potential harm \citep{Prorok2000,Gohagan2000}. Blood-based biomarker panels are attractive because they are minimally invasive, repeatable, and can exploit an individual's longitudinal history rather than a single cross-sectional snapshot.

Parametric Empirical Bayes (PEB) provides a pragmatic foundation for individualized screening thresholds by borrowing strength from a healthy reference population to stabilize person-specific estimates and control false positives under biological and technical variability \citep{robbins1992empirical,efron1973stein,carlin2000bayes,mcintosh2003parametric}. By comparing a subject's current measurement with the distribution expected from that subject's own history, PEB adapts the decision rule to between-person heterogeneity and within-person trajectories, and it improves on fixed single-threshold (ST) rules that ignore history \citep{mcintosh2003parametric}. There are existing studies on \rev{lung}, pancreatic, and l iver cancers for improved risk assessment using \rev{the} PEB algorithm \citep{fahrmann2025lead, irajizad2024biomarker, tayob2016improved}\rev{.}

Despite its success, current implementations of PEB have two important limitations. First, PEB does not explicitly account for the time elapsed between consecutive biomarker measurements. Consequently, an identical deviation from an individual's expected biomarker value is interpreted the same whether the previous sample was obtained three months, six months, or one year earlier, even though these intervals correspond to markedly different biological expectations for disease progression \citep{pepe2003statistical,JanesPepe2009,HeagertyLumleyPepe2000}. Second, when multiple biomarkers provide complementary information, PEB relies on predefined biomarker panels with fixed combination rules and lacks a principled framework for feature selection or learning an optimal composite score. As a result, biomarker weights cannot be optimized for specific clinical objectives, such as maximizing sensitivity at a fixed specificity, extending diagnostic lead time, or reducing unnecessary referrals \citep{pepe2003statistical,Hand2009}.

To address these limitations, we developed improved Parametric Empirical Bayes (iPEB), a generalized framework for longitudinal biomarker analysis that jointly models biomarker trajectories, performs objective-aligned feature selection, and learns optimized biomarker combinations. Unlike conventional PEB, iPEB explicitly incorporates the time elapsed between consecutive measurements, allowing prediction uncertainty to increase as the interval between visits grows while accounting for subject-specific longitudinal trends and relevant covariates. This enables more accurate interpretation of longitudinal biomarker changes by distinguishing expected biological variation from disease-related deviations. Rather than relying on predefined biomarker panels and fixed combination weights, iPEB optionally performs feature selection and optimizes biomarker weights directly according to the desired clinical objective. Specifically, iPEB can be optimized to maximize sensitivity at a prespecified specificity \citep{ghasemi2025novel,khoshfekr2025smags}, maximize diagnostic lead time, or jointly balance sensitivity and lead time subject to constraints on specificity or referral rate. The framework is modular and generalizes conventional PEB, with conventional PEB incorporated as a special case within the broader iPEB framework.

Because PEB already \snew{outperforms} \rev{the ST rule} \citep{fahrmann2025lead, irajizad2024biomarker, tayob2016improved}\rev{,} we frame all combinations as iPEB against PEB applied to a clinically established biomarker panel: the operative question is whether objective-driven multi-marker weighting improves on a strong PEB baseline. 

The remainder of the manuscript is organized as follows.
Section~\ref{sec:method} develops the method: the time-gap--aware layer (Section~\ref{sec:layer}), the multi-marker composite and its calibrated per-visit threshold (Section~\ref{sec:composite}), and the objective-driven weight optimization (Section~\ref{sec:weights}). Section~\ref{sec:theory} states the theoretical properties (calibration, minimum-variance prediction, and a power advantage), with proofs in Web Appendix~A. Section~\ref{sec:sim} reports a simulation study and Section~\ref{sec:app} an application to the Prostate, Lung, Colorectal, and Ovarian (PLCO) lung cohort. Section~\ref{sec:disc} concludes.

\section{The iPEB Method}
\label{sec:method}

\rev{Operationally, iPEB layers three user-selectable choices on the objective-driven weighting common to every setting: a random-slope or intercept-only mean model (chosen by the visit count), AR(1)/OU or i.i.d.\ innovations (chosen by whether serial correlation or drift is plausible), and an optional feature-selection step to a target panel size; the weighting itself has a scalar and a multivariate variant, selected automatically on validation.}

\subsection{Notation and the marginal PEB baseline}
\label{sec:notation}
We observe longitudinal panels on $n$ individuals; index subjects by $i=1,\dots,n$, biomarkers by $j=1,\dots,p$, and visits by $t\in\mathcal T_i$. Let $Y_{ijt}$ be biomarker $j$ for subject $i$ at visit $t$, let $\mathcal H$ denote the healthy reference set (no diagnosis during follow-up), and let $\mathcal C$ denote the cases, with diagnosis time $T_i$. Standard marginal PEB posits, for $i\in\mathcal H$,
\begin{equation}
Y_{ijt}\mid b_{ij}\sim\mathcal N(\mu_j+b_{ij},\,\sigma_j^2),\qquad b_{ij}\sim\mathcal N(0,\tau_j^2),
\label{eq:peb-base}
\end{equation}
with $\theta_j=(\mu_j,\sigma_j^2,\tau_j^2)$ estimated on $\mathcal H$, and forms the history-adjusted $z$-score $R_{ijt}=\{Y_{ijt}-\hat\mu_j-\hat b_{ij}(t)\}/\hat\sigma_j$ under i.i.d.\ residuals. Neither the mean adjustment nor the variance depends on the time gap between visits. \snew{By referencing each measurement against the subject's own history, PEB improves on a fixed single-threshold rule applied to the raw marker \citep{mcintosh2003parametric}.} iPEB relaxes both restrictions of \eqref{eq:peb-base}: it makes the residual structure time-gap--aware, and it replaces the single-marker score with an objective-driven composite.

\subsection{Time-gap--aware PEB layer}
\label{sec:layer}
For subject $i$, visit $t$, and marker $j$, let $h_{it}$ be a \rev{centered time index (visit time in months, centered so that the random intercept is the marker level at the reference time)} and $x_{it}$ a vector of covariates at visit~$t$ \rev{(here age, assay batch, and calendar year; these generally vary across visits, and a time-invariant covariate is simply the special case of a column that is constant in $t$)}. On healthy training subjects we posit the working model
\begin{equation}
Y_{ijt}=\mu_j+x_{it}^\top\beta_j+b^{(0)}_{ij}+b^{(1)}_{ij}h_{it}+\varepsilon_{ijt},
\label{eq:layer}
\end{equation}
with random effects $(b^{(0)}_{ij},b^{(1)}_{ij})^\top\sim\mathcal N(0,D_j)$ and residuals $\varepsilon_{ijt}$ following either an \snew{AR(1) (first-order autoregressive)} process with per-month correlation $\phi_j$ (effective correlation $\phi_j^{\Delta}$ over a gap $\Delta$) or a continuous-time \snew{Ornstein--Uhlenbeck (OU)} process with rate $\lambda_j$ (correlation $e^{-\lambda_j\Delta}$)\snew{; the OU form is the general continuous-time process, with AR(1) its special case on an equally spaced grid}. Fitting \eqref{eq:layer} on healthy training data by restricted maximum likelihood or a state-space filter \rev{first removes the fixed covariate effects $x_{it}^\top\beta_j$ (here age, batch, and calendar year) and the subject's own history to form} the one-step innovation
\begin{equation}
\nu_{ijt}=Y_{ijt}-\hat\mu_j-x_{it}^\top\hat\beta_j-\widehat{b^{(0)}_{ij}}(t)-\widehat{b^{(1)}_{ij}}(t)\,h_{it},
\label{eq:innov}
\end{equation}
whose prediction variance \emph{increases with the gap} $\Delta_{it}$ since the previous visit,
\begin{equation}
S_{ijt}=\sigma_j^2\bigl(1-\phi_j^{2\Delta_{it}}\bigr)\ \ \text{(AR(1))},\qquad
S_{ijt}=\sigma_j^2\bigl(1-e^{-2\lambda_j\Delta_{it}}\bigr)\ \ \text{(OU)}.
\label{eq:S}
\end{equation}
\rev{Equation~\eqref{eq:S} shows the gap-dependent \emph{residual} innovation term; the full one-step prediction variance additionally carries the state-uncertainty contribution $Z_{it}P^{-}_{it}Z_{it}^\top$ of the Kalman/GLS recursion (Web Appendix~B), which becomes negligible once the subject-specific effects are well estimated---the intercept-only regime of our application---so \eqref{eq:S} is the operative gap-dependent factor. Standardization and the calibration guarantee below use this full prediction variance, exactly as implemented.} The standardized score $R_{ijt}=\nu_{ijt}/\sqrt{S_{ijt}}$ down-weights longer gaps and up-weights recent information, so that ``today versus three months ago versus one year ago'' is judged on a comparable scale. \rev{\snew{We refer to this construction---the random intercept $b^{(0)}_{ij}$ and optional random slope $b^{(1)}_{ij}$ of \eqref{eq:layer}, together with the gap-scaled innovation variance \eqref{eq:S}---as the time-gap--aware \emph{layer}.} In practice the layer exposes two independent, user-selectable switches. The first is the mean model: it is intercept-only when most subjects contribute only a few visits (roughly one to four) and gains a random slope $b^{(1)}_{ij}h_{it}$ once most subjects have five or more visits, the point at which a subject-specific trend becomes reliably estimable. \snew{The figure of five visits is a practical rule of thumb rather than a sharp cutoff: with only a few visits the slope is barely separable from noise, so its random-effect BLUP shrinks essentially to zero and the model reduces to intercept-only regardless, which is why the precise threshold---three versus five versus ten---is not critical.} The second is the innovation model: the score is standardized by the gap-aware AR(1)/OU variance \eqref{eq:S} whenever serial correlation or longitudinal drift is plausible, and otherwise reduces to the i.i.d.\ case; because \eqref{eq:S} is evaluated at the actual gap $\Delta_{it}$, unequally spaced visits are handled automatically (the OU form is the continuous-time limit that applies when the gaps differ across visits). The two switches are set independently.} We bound $|\phi_j|\le0.8$ (or shrink $\lambda_j$) for numerical stability; the \snew{$p\times p$} healthy cross-marker covariance $\Sigma_R=\mathrm{Cov}_{\mathcal H}(R)$ \snew{(equivalently the cross-marker correlation, since the $R_{ijt}$ are standardized to unit variance)} is estimated \amk{by its healthy sample covariance, regularized by a small ridge on the diagonal for numerical stability}. The state-space (Kalman) recursions that produce $\nu_{ijt}$ and $S_{ijt}$ under irregular spacing are given in Web Appendix~B.

\subsection{Multi-marker composite and calibrated per-visit thresholds}
\label{sec:composite}
For a weight vector $w=(w_1,\dots,w_p)^\top$ with $\|w\|_2=1$, the composite score is $S_i(t;w)=\sum_j w_j R_{ijt}$. Under the healthy Gaussian working model $S_i(t;w)\sim\mathcal N(0,\,w^\top\Sigma_R w)$, giving a per-visit threshold
\begin{equation}
c(w,\alpha)=z_{\alpha}\sqrt{w^\top\Sigma_R w},
\label{eq:thr}
\end{equation}
so that a visit is flagged when $S_i(t;w)>c(w,\alpha)$. Because \eqref{eq:thr} rescales with $\sqrt{w^\top\Sigma_R w}$, the target specificity $\alpha$ is maintained for every visit and every gap (Proposition~\ref{prop:calib}). \snew{The threshold~\eqref{eq:thr} does not depend on $i$ or $t$: it is common across subjects and visits, and the per-subject, per-visit adaptivity is carried by the standardized score $R_{ijt}=\nu_{ijt}/\sqrt{S_{ijt}}$, so a single cutoff holds the target specificity for every subject and visit.} \amk{Expression~\eqref{eq:thr} is the threshold implied by the Gaussian working model; in practice, and throughout our simulations and application, we set the operating cutoff distribution-free as the empirical $\alpha$-quantile of the composite scores at the healthy \emph{training} visits, which coincides with \eqref{eq:thr} when the working model holds and is what the accompanying software implements.} \rev{Here $\Sigma_R$ is the pooled healthy innovation \emph{correlation}, which we take to be stable across visits so that it also serves as the per-visit correlation; when this correlation drifts with the gap, per-visit specificity holds on average and closely---though not exactly---per visit, as monitored by the stratified-specificity diagnostics of Web Appendix~A.}

The composite may equivalently be built by combining markers on the innovation scale before standardizing. Writing $\nu_t\in\mathbb R^p$ for the stacked innovations with covariance $S_t$ and $D_t=\mathrm{diag}(S_{t,11}^{1/2},\dots,S_{t,pp}^{1/2})$, so that $R_t=D_t^{-1}\nu_t$, the standardized composite $Z_R(t;w)=w^\top R_t/\sqrt{w^\top D_t^{-1}S_tD_t^{-1}w}$ equals $Z_C(t;a)=a^\top\nu_t/\sqrt{a^\top S_t a}$ for $a=D_t^{-\top}w$. Hence ``standardize-then-combine'' and ``combine-then-standardize'' yield identical decisions at any fixed $\alpha$; iPEB realizes this through two operational variants (Section~\ref{sec:weights}).

\subsection{Objective-driven weight optimization}
\label{sec:weights}
Let $W$ denote a pre-diagnosis window. Weights are chosen by minimizing, on a tuning split, the scalarized clinical loss
\begin{equation}
L(w)=-\operatorname{Sens}_W(w)+\lambda_1\{\alpha-\operatorname{Spec}(w)\}_{+}+\lambda_2\operatorname{Ref}(w)-\lambda_3\operatorname{LT}(w),
\label{eq:loss}
\end{equation}
where $\operatorname{Spec}(w)$ is the per-visit specificity among healthy visits; $\operatorname{Sens}_W(w)$ is the probability that a case is flagged at least once within its pre-diagnosis window $[T_i-W,T_i)$; $\operatorname{Ref}(w)=1-\operatorname{Spec}(w)$ is the per-visit referral (false-positive) rate; and $\operatorname{LT}(w)$ is the median lead time among detected cases. The penalties instantiate three clinical \emph{objectives}: \emph{sensitivity} (referral-controlled sensitivity, $\lambda_3=0$), \emph{lead time} ($\lambda_3>0$, rewarding earlier detection), and a \emph{combined} objective that additionally penalizes referral ($\lambda_2>0$). Because the loss retains the sensitivity term $-\operatorname{Sens}_W$ and the specificity floor under every objective, the lead-time objective extends detection lead \emph{while maintaining sensitivity} rather than maximizing lead time in isolation---which alone would trivially favor a single late-detected case; the loss and its behavior are detailed in Web Appendix~C. \rev{The practitioner selects one of these three objectives---the default is sensitivity---rather than tuning penalty weights: each objective corresponds to a \emph{fixed} weight profile, $(\lambda_1,\lambda_2,\lambda_3)=(50,0,0)$ for sensitivity, $(50,0,0.7)$ for lead time, and $(50,1,0.5)$ for the combined objective (with $\operatorname{LT}$ entered on a scale normalized by a two-year reference). Fixing the profiles keeps the input interpretable---a clinical goal rather than three opaque constants.} Because the threshold \eqref{eq:thr} is recalibrated to $\alpha$ for every candidate $w$, the specificity penalty is essentially inactive and $\lambda_3$ is the operative lead-time lever; ties are broken by preferring higher $\operatorname{Sens}_W$, then lower referral, then sparser weights. \rev{The window $W$ is an optional hyperparameter of the objective; by default it spans the entire pre-diagnostic period, so detection is scored over the whole trajectory---a case counts as detected if it crosses at any pre-diagnostic visit---in \emph{both} weight optimization and evaluation. All results in Sections~\ref{sec:sim}--\ref{sec:app} use this default, so no fixed detection window is imposed at any stage.}

Under the Gaussian working model the loss admits a closed-form minimizer. On the innovation scale a healthy visit has $R_{it}\sim\mathcal N(0,\Sigma_R)$ and a case in window has $R_{it}\sim\mathcal N(\mu_\Delta,\Sigma_R)$, where $\mu_\Delta$ is the case-versus-healthy mean shift. For $S_i(t;w)=w^\top R_{it}$ the detection probability at specificity $\alpha$ is $1-\Phi\{z_{\alpha}-\Lambda(w)\}$ with noncentrality $\Lambda(w)=w^\top\mu_\Delta/\sqrt{w^\top\Sigma_R w}$, so maximizing power is maximizing $\Lambda(w)$; the maximizer of this generalized Rayleigh quotient is
\begin{equation}
w^\star\propto\Sigma_R^{-1}\mu_\Delta,
\label{eq:wstar}
\end{equation}
normalized so that $w^{\star\top}\Sigma_R w^\star=1$ (and projected onto the nonnegative orthant if required). iPEB realizes \eqref{eq:wstar} through two variants that share this loss and thresholding and correspond to the two equivalent orderings of Section~\ref{sec:composite}. The \emph{scalar} combiner (iPEB-S) first standardizes each marker by its own innovation variance, $R_{ijt}=\nu_{ijt}/\sqrt{S_{ijt}}$, and forms
\begin{equation}
Z_{\mathrm S}(t;w)=\frac{w^\top R_{it}}{\sqrt{w^\top\Sigma_R w}},\qquad w\propto\Sigma_R^{-1}\mu_\Delta,
\label{eq:ipebS}
\end{equation}
which needs only marginal innovation variances and the healthy correlation $\Sigma_R$, and is therefore cheap and robust. The \emph{multivariate} combiner (iPEB-M) instead whitens the markers jointly on the innovation scale,
\begin{equation}
Z_{\mathrm M}(t;a)=\frac{a^\top\nu_{it}}{\sqrt{a^\top S_t a}},\qquad a\propto S_t^{-1}m_\Delta,\quad m_\Delta=D_t\mu_\Delta,
\label{eq:ipebM}
\end{equation}
where $m_\Delta=D_t\mu_\Delta$ is the case mean shift of $\nu_{it}$ on the raw innovation scale; it exploits the full cross-marker innovation covariance $S_t$ but requires its stable estimation. By the equivalence of Section~\ref{sec:composite} the two coincide when $S_t$ is correctly specified ($a=D_t^{-\top}w$); in finite samples they differ, so iPEB fits both variants on the training data, \snew{retains the one with the higher validation sensitivity at the target specificity---a single variant-selection criterion applied across all three objectives---}and refits it on the full training set before locking it. \snew{This criterion is deliberate rather than incidental. The two combiners are population-equivalent and differ only in how robustly they estimate the innovation covariance in finite samples, so choosing between them is an estimation-robustness decision rather than an objective-alignment one---the objective is already imposed when each candidate's weights are fitted. We select on sensitivity at the target specificity because it is computed from all cases and is therefore low-variance, whereas lead time on the small validation slice is estimated only from the few detected cases and is too noisy to select on reliably.} Weights are signed---a marker may enter negatively when informative only after adjustment---and unit-normalized; nonnegativity is optional. When empirical clinical metrics are optimized directly the objective is non-convex, and we use scalarization of \eqref{eq:loss} with the tie-breakers above, or population-based search for a Pareto frontier \citep{Deb2002}; a multi-window robust extension is a second-order cone program and is derived in Web Appendix~C. \rev{When a smaller panel is preferred, iPEB optionally performs \emph{feature selection} by objective-driven backward elimination: starting from all markers, at each step it drops the marker whose removal least degrades the validation objective, continuing until a target panel size is reached, and then re-optimizes the weights on the retained markers; the number of markers to keep is a user choice. We illustrate this in the selected-four comparison of Section~\ref{sec:app}, where iPEB reduces six markers to four.} The complete procedure is summarized in Algorithm~1.

\begin{center}
\fbox{\begin{minipage}{\linewidth}\setstretch{1}
\textbf{Algorithm 1 (iPEB).}
\emph{Required input:} longitudinal panels $Y_{ijt}$, labels $(\mathcal H,\mathcal C)$ and diagnosis times $T_i$; a target specificity $\alpha$; a clinical objective (sensitivity, lead time, or combined); and a train/tune/test split. \rev{\emph{Optional input:} a detection window $W$ (default: the whole pre-diagnostic trajectory), a target panel size $q<p$ (which turns on feature selection); and the layer configuration---a random slope on or off, and i.i.d.\ versus AR(1)/OU innovations (Section~\ref{sec:layer}). The chosen objective fixes the penalty profile $(\lambda_1,\lambda_2,\lambda_3)$ (Section~\ref{sec:weights}).}\\
\emph{Step 1 (layer).} On healthy training subjects, fit \eqref{eq:layer} per marker; compute innovations $\nu_{ijt}$ and gap-dependent variances $S_{ijt}$ from history $\le t-1$; form $R_{ijt}$ and estimate $\Sigma_R$ \amk{by its regularized healthy sample covariance}.\\
\emph{Step 2 (score and threshold).} For candidate $w$, set $S_i(t;w)=\sum_j w_jR_{ijt}$ and threshold $c(w,\alpha)=z_{\alpha}\sqrt{w^\top\Sigma_R w}$.\\
\emph{Step 3 (optimize weights).} On the tuning split, evaluate $\operatorname{Sens}_W$, specificity, referral, and lead time; minimize \eqref{eq:loss} (closed form \eqref{eq:wstar}, or search) for both the scalar and the multivariate variant.\\
\rev{\emph{Step 3a (optional feature selection).} If a target size $q<p$ is requested, backward-eliminate markers---at each step dropping the marker whose removal least degrades the validation objective---until $q$ remain, then re-optimize the weights on the retained markers.}\\
\emph{Step 4 (select and lock).} \snew{Keep the variant with the higher validation $\operatorname{Sens}_W$ (sensitivity at the target specificity), used as a common variant-selection criterion across all objectives}; refit on the full training set and lock the weights, panel, and threshold.\\
\emph{Step 5 (evaluate).} Evaluate once on the held-out test data; \rev{optionally summarize agreement between methods by a subject-level cross-classification of detected cases, or attach uncertainty by a subject-level bootstrap.}
\end{minipage}}
\end{center}

With an intercept-only i.i.d.\ layer and a single-marker weight, Algorithm~1 reduces exactly to marginal PEB.

\section{Theoretical Properties}
\label{sec:theory}
\amk{The three propositions below are, respectively, classical calibration, minimum-variance prediction, and Neyman--Pearson optimality results; the contribution is that each continues to hold for the time-gap--aware layer under irregular visit spacing.} We state analogues, for the time-gap--aware layer, of the marginal-PEB properties of \citet{mcintosh2003parametric}. Throughout, $R_{ijt}=\nu_{ijt}/\sqrt{S_{ijt}}$ is the standardized innovation from the Kalman filter/GLS, $\Sigma_R$ is the healthy cross-marker covariance, $c(w,\alpha)=z_{\alpha}\sqrt{w^\top\Sigma_R w}$, and $\mathcal F_{t-1}$ is the history up to the previous visit. Proofs are in Web Appendix~A.

\begin{proposition}[Calibration]
\label{prop:calib}
Under the healthy model, for any $i,t$ and gap $\Delta_{it}>0$, $\mathbb E_{\mathcal H}[R_{ijt}\mid\mathcal F_{t-1}]=0$, $\mathrm{Var}_{\mathcal H}(R_{ijt}\mid\mathcal F_{t-1})=1$, and $R_{ijt}\sim\mathcal N(0,1)$. Consequently $\Pr_{\mathcal H}\{S_i(t;w)\le c(w,\alpha)\}=\alpha$ for every $t$ and every $\Delta_{it}$; that is, the per-visit specificity equals $\alpha$ (and the per-visit false-positive rate $1-\alpha$).
\end{proposition}
\noindent\emph{Proof idea.} The Kalman prediction-error decomposition gives $\nu_{ijt}\mid\mathcal F_{t-1}\sim\mathcal N(0,S_{ijt})$ with $S_{ijt}$ a deterministic function of the gap $\Delta_{it}$; standardizing yields $\mathcal N(0,1)$, and linearity of the composite gives the specificity claim. Full proof in Web Appendix~A.

\begin{proposition}[Minimum-variance prediction]
\label{prop:mvp}
The \rev{time-gap--aware} empirical Bayes (Kalman/BLUP) one-step predictor $\widehat m_{ijt}$ minimizes prediction mean-squared error among linear unbiased predictors using the same history; hence the numerator of $R_{ijt}$ has minimum variance in that class. \rev{The gap-ignorant mean-type baselines---the subject mean, last observation carried forward, or the i.i.d.\ empirical Bayes predictor---are (marginally) linear predictors built on a misspecified or degenerate covariance, and incur no smaller prediction MSE in the model-based settings we consider.}
\end{proposition}
\noindent\emph{Proof idea.} Given the correct second-order structure, the Kalman/BLUP predictor is the generalized-least-squares projection of the new observation onto the history, which is minimum-variance linear unbiased by Gauss--Markov; gap-ignorant mean-type predictors use a misspecified (or degenerate) covariance and so cannot improve on it. Full proof in Web Appendix~A.

\begin{proposition}[Power advantage over gap-ignorant rules]
\label{prop:power}
For a case alternative that shifts the mean by $\delta_j>0$ on a marker subset within the pre-diagnostic window \rev{while leaving the healthy innovation covariance $\Sigma_R$ unchanged}, and fixed per-visit specificity $\alpha$ for every gap, the Neyman--Pearson most powerful test thresholds a linear form of $R_{ijt}$ (univariate $R_{ijt}$; multivariate $w^\top R_{it}$ with $w\propto\Sigma_R^{-1}\mu_\Delta$). Any statistic that does not standardize by $\sqrt{S_{ijt}}$ cannot be most powerful for all $\Delta_{it}$ simultaneously; hence the time-gap--aware rule has weakly higher true-positive rate at fixed specificity, and strictly higher whenever $\Delta_{it}$ varies.
\end{proposition}
\noindent\emph{Proof idea.} For each fixed gap the common-covariance Gaussian innovation family has monotone likelihood ratio, so Neyman--Pearson thresholds $w^\top R_{it}$ with $w\propto\Sigma_R^{-1}\mu_\Delta$ \rev{(for a standardized shift $\mu_\Delta$ that is stable across gaps)}; a rule that omits the factor $1/\sqrt{S_{ijt}}$ is a different, suboptimal linear statistic for at least one gap and cannot be most powerful at all gaps at once, so averaging over the empirical gap distribution the standardized rule dominates. Full proof in Web Appendix~A.

Proposition~\ref{prop:calib} guarantees that the operating specificity is honored under irregular schedules; Proposition~\ref{prop:mvp} identifies the layer as optimal among linear predictors; and Proposition~\ref{prop:power} shows that ignoring the gap forfeits power. The results are model-based (Gaussianity, correctly specified mean and variance structure); under heavy tails or change-points, calibration holds approximately and the advantages attenuate, so innovation quantile--quantile plots, residual autocorrelation, and stratified-specificity diagnostics are recommended (Web Appendix~A).

\section{Simulation Studies}
\label{sec:sim}
\amk{We evaluate iPEB in two complementary simulation studies. The first vets the robustness of iPEB across varying dimension, correlation, drift sparsity, and distributional misspecification, and their combinations. The second quantifies the advantages of iPEB through mechanistic checks that isolate the benefits of objective-driven weighting, lead-time optimization, and the time-gap--aware layer.}

\subsection{\amk{Simulation study designs}}
\label{sec:sim-design}

\smallskip\noindent\textbf{Factorial study.} We simulated longitudinal panels with healthy random effects and case-specific preclinical drift---an unequal-magnitude linear ramp over the final months before diagnosis on a subset $\mathcal D$ of ``driver'' markers---so that cases and controls share an identical baseline distribution and diverge only within the window, mirroring a screening cohort. \amk{We varied the design across four main factors: (i) panel dimension ($p\in\{5,10\}$); (ii) correlation strength between markers (low $0.2^{|j-k|}$, high $0.7^{|j-k|}$, or two-block); (iii) drift sparsity, the proportion of signal-carrying markers ($|\mathcal D|/p\in\{0.4,0.6\}$, with unequal effects); and (iv) distributional misspecification (none, $t_5$, AR(1), or log-normal). A full factorial over these axes yields 48 unique settings, each averaged over 100 Monte Carlo cohorts. Each simulated cohort comprises $800$ controls and $200$ cases.} \rev{Within every cohort, the iPEB weights, the comparators, and the specificity threshold were learned on a random \amk{$75\%$} training sample of subjects (cases and controls)---within which a further $25\%$ validation slice selected the scalar-versus-multivariate variant---and all metrics evaluated on the held-out \amk{$25\%$}, so no subject used for fitting was reused for evaluation.} Comparators---the best marginal PEB marker, the equal-weight composite, and the first principal component (PCA-1)---used identical PEB thresholding, isolating the effect of weighting. \rev{Because each subject contributes six visits (the $\ge5$-visit regime), iPEB and all comparators use a random-slope layer throughout the factorial, and visit times are drawn irregularly per subject as a mixture of semiannual and annual gaps; the time-gap--aware AR(1)/OU innovation is enabled only in the AR(1)-misspecification cells---where the generative process carries serial correlation---and reduces to i.i.d.\ elsewhere. The same layer is shared by iPEB and every comparator, so the study still isolates the objective-driven weighting. The simulated panels carry no covariates, so the covariate-adjustment term $x_{it}^\top\beta_j$ is exercised only in the real-data application.} \snew{These axes span regimes that favor and disfavor a fixed composite alike, rather than only settings chosen to suit iPEB:} higher dimension and marker correlation force a combiner to balance redundant signal, which whitening by $\Sigma_R^{-1}$ exploits, while the misspecification axis probes robustness of the Gaussian working model, with log-normal noise the adversarial regime in which a single dominant marginal marker can outperform any composite. Because the factorial imposes \emph{synchronous} drift across markers, it carries no cross-marker timing heterogeneity for a combiner to exploit, and median lead time among detected cases is confounded by detection rate; lead time is therefore assessed in the differential-timing scenario below, and the factorial reports \rev{per-patient sensitivity under whole-trajectory detection} (iPEB under the sensitivity objective) at $95\%$ \rev{per-visit} specificity. Full design details and the per-cell tables are in Web Appendix~D and Web Tables~1--2.

\smallskip\noindent\textbf{Mechanism study.} \rev{\snew{Four} focused scenarios, each over 200 cohorts at $95\%$ specificity, isolate the mechanisms (Table~\ref{tab:focused}, Figures~\ref{fig:s2}--\ref{fig:s3}; the Simulation~1 mean ROC is Web Figure~1\snew{, and the Simulation~4 sensitivity comparison Web Figure~2}); detection uses the whole pre-diagnostic trajectory (Section~\ref{sec:app-design}). \rev{S1 and S2 place every subject on a regular visit grid (eight and ten visits) and, being in the $\ge5$-visit regime, fit a random slope with i.i.d.\ innovations---these designs carry no longitudinal trend or serial correlation, so AR(1)/OU is not applied---and the comparison isolates the objective-driven \emph{weighting}, since iPEB and the fixed panel share the same layer. S3 keeps the regular grid but injects genuine healthy drift and serial correlation, and contrasts an intercept-only i.i.d.\ layer against the random-slope-plus-OU layer to isolate the \emph{layer} itself; each subject contributes eight visits so the slope and OU rate are well estimated.}} \snew{Scenario~4 repeats this S3 mechanism (healthy drift with serial correlation) on an \emph{irregularly} spaced visit grid, with inter-visit gaps ranging from about three months to two years\amk{, contrasting the gap-ignorant PEB (intercept-only, i.i.d.) with the gap-aware iPEB (random slope, OU)}.}

\subsection{\amk{Results}}
\label{sec:sim-results}

\smallskip\noindent\textbf{Factorial study.} \rev{Across the 48-cell factorial \snew{simulation} (Figure~\ref{fig:ext}), iPEB attained the highest mean sensitivity overall (\amk{$0.342$} versus $0.316$ best-marginal, $0.308$ equal-weight, $0.305$ PCA-1), exceeding the best single-marker benchmark in $39$ of $48$ cells, the equal-weight composite in \amk{$45$}, and \amk{PCA-1 in $45$}. \amk{Throughout, a method ``wins'' a cell if its mean sensitivity over that cell's 100 cohorts is strictly higher than the comparator's.} Margins widened under correlation and higher dimension---where fixed linear combinations are least efficient (for example, $p=10$, high correlation, sparsity $0.6$, no misspecification: iPEB \amk{$0.372$} versus $0.304$ equal-weight, $0.307$ PCA-1, $0.318$ best-marginal). All $9$ cells in which iPEB did not exceed the best single marker were log-normal cells; but even under this adversarial multiplicative skew iPEB still \emph{won} $3$ of the $12$ log-normal cells (all at higher dimension), where the discriminating signal is spread across enough markers that combining still helps. Where a single dominant marker carries the extreme values, that marker can still edge out any composite; we report rather than suppress those cells, as they mark the honest boundary of a multi-marker approach. Per-cell means and standard deviations are given in Web Tables~1--2.}

\smallskip\noindent\textbf{Mechanism study.} \rev{Under \emph{marker heterogeneity} (S1: two case subgroups driven by different markers; sensitivity objective), objective-driven weighting improves sensitivity over the history-adjusted composite by $+0.094$ ($0.923$ versus $0.829$) and raises the AUC from $0.972$ to $0.984$. Under \emph{differential timing} (S2: an early-rising marker whose signal begins about two years before diagnosis and a late-rising marker that rises about three months before, a fixed panel developed on an independent cohort; lead-time objective; Figure~\ref{fig:s2}), iPEB detects roughly $0.94$ years earlier than the fixed panel (median lead $1.19$ versus $0.25$ years) while lifting sensitivity to nearly one ($0.998$ versus $0.844$): a single fixed panel is forced to favor either the early marker (long lead, lower sensitivity) or the late marker (short lead, higher sensitivity), whereas iPEB combines both. Under \emph{longitudinal drift with serial correlation} (S3: healthy random slope with AR(1) residuals; sensitivity objective), an intercept-only i.i.d.\ layer forfeits sensitivity ($0.529$; i.i.d.\ iPEB $0.518$), whereas the random-slope-plus-OU layer restores it to $0.846$ ($+0.317$; AUC $0.880\!\to\!0.974$); because weighting is held fixed across these rows, the gain is attributable to the time-gap--aware layer itself.} \snew{\amk{On the irregular grid of S4,} unable to rescale for the varying gaps, the gap-ignorant PEB detects only $0.275$ of cases, whereas the gap-aware iPEB detects $0.892$ (AUC $0.978$ versus $0.785$). Because the same gap-ignorant rule reaches $0.529$ on the regular S3 grid, the further drop to $0.275$ isolates the additional cost of ignoring irregular spacing (Web Figure~2).}

\section{Application to the PLCO Lung Cohort}
\label{sec:app}

\subsection{Cohort and comparisons}
\label{sec:app-design}
The Prostate, Lung, Colorectal, and Ovarian (PLCO) Cancer Screening Trial is a randomized, multicenter U.S.\ trial; blood specimens were collected at baseline and annually until lung-cancer diagnosis or end of follow-up \citep{Prorok2000,Gohagan2000,m2015plco,fahrmann2022blood,irajizad2023mortality}. \rev{We analyzed screening-eligible ever-smokers with at least ten pack-years, imposing no smoking-cessation restriction. A subject is included as a \emph{case} if at least one of its blood draws falls within two years of the lung-cancer diagnosis; all draws from included subjects are retained, so a case contributes its entire longitudinal trajectory rather than a single windowed draw. Subjects with at least two pre-diagnostic draws were retained, and the analysis cohort comprised 1{,}998 subjects (324 cases, 1{,}674 controls) with a median of four draws per subject; characteristics are summarized in Table~\ref{tab:char}.} Covariates were \rev{age (biological aging, which shifts baseline marker levels), assay batch (technical variation across laboratory processing batches), and calendar year (secular drift in the assay and study population)}, and the innovation variance used the AR(1)/OU gap-scaling of \eqref{eq:S}; \rev{we used an intercept-only mean model, since the median of four visits per subject falls in the intercept-only regime of Section~\ref{sec:layer} (a subject-specific slope is reliably estimable only with more visits). The random slope and the gap-aware AR(1)/OU innovation variance are separable options, so we retain the gap-scaling of \eqref{eq:S}---keeping the layer's time-gap awareness active on this cohort---and omit only the random slope---the two switches are set independently (Section~\ref{sec:layer}).} \rev{To probe transportability across sites rather than reuse patients across folds, we used a single fixed \emph{center-based} split: the six centers $\{2,4,5,6,8,9\}$ (1{,}402 subjects) formed the training set and the four held-out centers $\{1,3,10,11\}$ (596 subjects, about $30\%$ of the cohort---mirroring the $30\%$ held out in each \amk{mechanism-scenario} cohort) the test set. Within the training centers, a $25\%$ validation slice (subjects held out from weight fitting) chose the scalar-versus-multivariate variant and drove the objective-driven feature selection. \rev{Across the nine training-set fits (three comparisons $\times$ three objectives), the scalar combiner (iPEB-S) was selected in seven and the multivariate combiner (iPEB-M) in two---the multivariate form winning under the combined objective for the six-marker panel and under the sensitivity objective for the selected-four panel; this selection uses only the training and validation data and never the test centers.} No subject appears on both sides, so fitting and evaluation never share data; we favor this clean separation over resampling and therefore report point estimates \amk{(a subject-level bootstrap of the held-out centers, holding the training-calibrated threshold fixed, would attach uncertainty to the reported gains; Algorithm~1, Step~5)}. The combiner was optimized on the training centers at the high-risk triage specificity $\alpha=0.60$ (the operating point for a smoking-eligible screening population) and evaluated once on the test centers at $\{0.60,0.95,0.99\}$. The closed-form weight \eqref{eq:wstar} is independent of $\alpha$ (which enters only the threshold \eqref{eq:thr}, not the direction), so weights learned at $\alpha=0.60$ stay power-maximizing at $0.95$ and $0.99$ and reporting at the stringent points is no operating-point mismatch. We designate sensitivity at $95\%$ specificity in the selected-four comparison as the primary endpoint and treat the remaining specificities, objectives, and comparisons as secondary.} \rev{Sensitivity and lead time are scored per patient on the \emph{whole} pre-diagnostic trajectory rather than within a fixed window: a case subject is detected if \emph{any} of its pre-diagnostic draws crosses the threshold, and its lead time is the interval from the \emph{earliest} crossing to diagnosis. Specificity is scored \emph{per visit}: the operating threshold is calibrated on the \emph{training} control \emph{visits} so that the target fraction (e.g., $95\%$) of control draws fall below it, and it is then applied unchanged to the test centers so no test information enters the threshold. This pairs a per-patient sensitivity with a per-visit specificity, matching the per-visit calibration of Proposition~\ref{prop:calib}; sensitivity, lead time, and AUC are then read off the held-out test centers at that fixed, training-calibrated threshold, so the specificity is set in-sample and every reported gain is out-of-sample. \snew{Because sensitivity is scored per patient while specificity is scored per visit, the reported AUC combines a per-patient and a per-visit level; we therefore read it as a descriptive summary of discrimination rather than attaching model-based standard errors to it, and anchor the primary comparison on sensitivity at a fixed specificity, following the operating-point reporting used for these panels \citep{irajizad2024biomarker}. \amk{Concretely, sweeping the common per-visit threshold traces the curve that pairs each induced per-visit false-positive rate with the resulting per-patient true-positive rate; the reported AUC is the area under this per-patient/per-visit curve, and any resampling-based uncertainty for it should therefore resample at the subject level.}}}

\rev{We report three like-for-like comparisons, each contrasting a \emph{frozen} composite (scored through PEB) with iPEB on a panel of the \emph{same} size, so that any difference reflects objective-driven weighting rather than marker count. In the \emph{four-marker} comparison the frozen panel is the published four-marker composite (4MP) of \citet{fahrmann2022blood}---a fixed, published weighted combination of the log-normalized markers CA125, CEA, CYFRA~21-1, and pro-surfactant protein~B (pro-SFTPB); the exact coefficients are published in full in \citet{fahrmann2022blood} and are applied here verbatim, so the frozen arm is fully reproducible from that reference (we do not reprint the combination, as it belongs to that published work)---and iPEB is applied to those same four markers. In the \emph{six-marker} comparison, absent a published six-marker panel, the frozen comparator is a logistic regression of case/control status on all six markers (CA125, CEA, CYFRA~21-1, pro-SFTPB, osteopontin (OPN), HE4), fit on the training centers and then frozen with PEB applied to its linear predictor---a fit-then-freeze logistic composite being the standard deployable rule when no established panel exists, since a fixed linear combiner is locked before use; iPEB is applied to the same six markers (representative coefficients in the caption of Table~\ref{tab:lung}). The third, \emph{selected-four} comparison lets each method choose its own four-marker panel from the six by its \emph{own} selection strategy, so the contrast is between the two selection philosophies at a common panel size. The frozen side follows the standard association-based route: it ranks the six markers by the magnitude of their logistic $z$-statistics on the training centers and retains the top four, then refits and freezes a logistic on those four (coefficients in the caption of Web Table~5). iPEB instead uses its objective-driven backward elimination (Step~3a of Algorithm~1) to reach four markers and re-optimizes on them. Crucially, both selections are done entirely on the training centers---the held-out centers are touched only at evaluation---so neither can leak test information, and the comparison is fair despite the differing selection criteria. Notably the logistic ranking recovered exactly the 4MP markers, whereas iPEB's objective-driven selection brought in osteopontin and HE4 in place of two 4MP members (which two depending on the objective)---the objective, not a fixed panel, drives its selection. Sensitivity is reported for iPEB under the sensitivity objective and lead time under the lead-time objective; full per-objective results on both the test and training centers appear in Web Tables~3--5 (the selected-four case in Web Table~5, whose caption also lists the marker set each objective selected). As an unweighted, subject-level companion to the tabulated gains we cross-classified the held-out case subjects by which method detected them at $95\%$ specificity (detected by both, by iPEB only, by the frozen panel only, or by neither), reported in Web Table~6.}

\subsection{Results}
\label{sec:app-results}
\rev{All numbers below are computed once on the held-out test centers (99 case and 497 control subjects; Table~\ref{tab:lung}). Under the sensitivity objective for optimization, at the permissive $60\%$ specificity, the operating threshold is low and both methods flag almost every case ($0.92$--$0.96$), so iPEB is marginally lower there; the separation appears where screening actually operates at high specificity. At $95\%$ specificity iPEB improved sensitivity over the frozen panel by $+0.051$ (four-marker, $0.485$ vs $0.434$), $+0.071$ (six-marker, $0.485$ vs $0.414$), and $+0.060$ (selected-four, $0.515$ vs $0.455$); at $99\%$ the gains were $+0.060$, $+0.080$, and $+0.020$, respectively. Additionally, iPEB's discrimination was slightly higher than the panel's in every comparison (test-set AUC $0.877$ vs $0.875$ four-marker, $0.876$ vs $0.869$ six-marker, and $0.873$ vs $0.871$ selected-four). Under the lead-time objective iPEB and PEB were comparable (Table~\ref{tab:lung}). Median lead times were long throughout (roughly $1$--$2.6$ years). iPEB's median lead was marginally shorter than the panel's at the lower specificities (by up to about $0.4$ years at $0.60$ and $0.95$) and comparable-to-longer at the stringent $99\%$ point; the small dip is expected rather than a loss of early detection---iPEB flags more cases, and the additional ones it catches tend to present closer to diagnosis, which pulls the median down. The multi-year lead is thus retained while sensitivity rises, so the sensitivity gain does not come at the cost of catching cases later. A subject-level cross-classification of the test cases at $95\%$ specificity agreed without weighting: $9$ versus $4$ cases were detected by iPEB alone versus the panel alone in the four-marker comparison, $11$ versus $4$ in the six-marker comparison, and $14$ versus $8$ in the selected-four comparison, so the net movement of individual patients favored iPEB (Web Table~6).}

\section{Discussion}
\label{sec:disc}
iPEB integrates a time-gap--aware PEB layer with objective-driven multi-marker weighting, ensuring that individualized, history-adjusted scores and their calibrated per-visit threshold account for longitudinal biomarker dynamics while directly aligning the resulting composite score with the desired clinical objective. The time-gap--aware standardization allows prediction uncertainty to increase with the elapsed time since the previous measurement, thereby preserving per-visit specificity under irregular sampling schedules (Proposition~\ref{prop:calib}). In parallel, objective-driven weighting prioritizes biomarkers that contribute most strongly to performance at the prespecified operating point, enabling both feature selection and optimized biomarker combination in accordance with the clinical objective (Proposition~\ref{prop:power}). In a 48-cell factorial spanning correlation, dimension, and misspecification, iPEB improved sensitivity over deployable composites in nearly all settings; in a dedicated differential-timing scenario it additionally delivered substantially earlier detection. \rev{On the PLCO lung cohort, evaluated on held-out screening centers, it improved high-specificity sensitivity over three like-for-like benchmarks---a published four-marker panel, a frozen six-marker logistic composite, and a marker set each method selected for itself---while discrimination was slightly higher for iPEB under the sensitivity objective and comparable or slightly higher under the lead-time objective, locating its advantage precisely where a screening program operates: more true cases flagged at a fixed low false-positive rate, with lead time preserved or extended at the stringent operating points. That the gains held on centers not used for fitting, and that a simple patient-level count of who was detected moved in the same direction as the tabulated sensitivities, argues that the improvement reflects the method rather than an artifact of a favorable split.}

Several limitations bound these gains. The approach targets biomarkers with a genuine \emph{within-subject} longitudinal signal. \snew{When the informative signal is instead essentially cross-sectional and concentrated in a single dominant marker, history adjustment can remove the very feature being screened, and multi-marker weighting then offers little over that marker alone. We regard this as a scope condition rather than a failure of the framework. It is consistent with the log-normal cells of the factorial study, the only regime in which the best single marker sometimes outperforms iPEB. Even there, iPEB still won $3$ of the $12$ log-normal cells once the dimension was high enough for combining to help.} Random slopes and autocorrelation parameters are difficult to estimate with only two or three visits per subject; intercept-only models with pooled temporal parameters are the pragmatic default. \snew{More generally, the advantage of the time-gap--aware layer grows with the number of serial draws per subject: the PLCO cohort's median of four draws is comparatively light, so the real-data gains reported here are, if anything, conservative relative to what denser longitudinal sampling would afford.} Even at high specificity, positive predictive value is bounded by the low prevalence of preclinical cancer in a screening population: because the great majority of those tested are disease-free, the fraction of flagged subjects who truly have cancer stays modest however favorable the sensitivity--specificity trade-off, so iPEB is best positioned as a first-stage triage that enriches a population for confirmatory work-up rather than as a standalone diagnostic. Reporting sensitivity at a fixed high specificity, as we do, isolates the operating characteristic the method controls and leaves positive predictive value to be read against the prevalence of the intended deployment setting. Abrupt assay changes likewise require refitting. External validation is necessary before transport to new populations, and prospective evaluation is the natural next step.


\backmatter

\section*{Acknowledgements}
The authors thank the PLCO participants and investigators. \rev{The authors used AI-assisted tools for language editing; all methods, results, and text were reviewed and verified by the authors.} \amk{This work was supported by the National Institutes of Health [U01CA271888 to S.H.]; the Cancer Prevention and Research Institute of Texas [RP160693 to K.A.D.]; and generous philanthropic contributions to The University of Texas MD Anderson Cancer Center Moon Shots Program.} \rev{The authors declare no potential conflicts of interest.}

\section*{Data Availability Statement}
PLCO biospecimen and biomarker data are available from the National Cancer Institute's Cancer Data Access System (CDAS) under its standard data-use agreement.

\section*{Supplementary Materials}
Web Appendices A--D, Web Figures~1 and~2, and Web Tables~1--6, referenced in Sections~\ref{sec:layer}--\ref{sec:app}, are available with this paper at the Biometrics website on Oxford Academic. iPEB is implemented in an open-source \textsf{R} package \rev{(\texttt{iPEB}), openly available at \texttt{https://github.com/bitansa/iPEB} with a CRAN release forthcoming}, providing the time-gap--aware PEB layer, the scalar and multivariate combiners, and the calibration, sensitivity, lead-time, and \rev{evaluation} utilities used here. \rev{A companion \textsf{R}/Shiny application ships with the package and can be launched with \texttt{iPEB::run\_app()}; a hosted, browser-based version for users without \textsf{R} \amk{will be} linked from the repository.} \rev{Scripts that} reproduce the simulation studies and the real-data workflow (the latter on a synthetic cohort, as the PLCO data are controlled-access) \rev{are available} at \texttt{https://github.com/bitansa/iPEB/tree/main/reproducibility}.

\bibliographystyle{biom}
\bibliography{bibliography}

\begin{table}[htbp]\centering
\caption{Mechanism simulation scenarios: mean (standard deviation) over 200 cohorts at $95\%$ \rev{per-visit} specificity. \rev{Sens, per-patient sensitivity under whole-trajectory detection}; Lead, median lead time (years). iPEB uses the sensitivity objective (S1, S3\snew{, S4}) and the lead-time objective (S2). \snew{S4 repeats the S3 mechanism on an irregularly spaced visit grid; its lead-time column is confounded by the large detection-rate gap and is not the comparison of interest.}}
\label{tab:focused}
\begin{tabular}{llccc}
\hline
Scenario & Method & AUC & \rev{Sens} & Lead (yr)\\
\hline
\multicolumn{5}{l}{\emph{S1: heterogeneous markers}}\\
 & PEB & \rev{0.972 (0.008)} & \rev{0.829 (0.050)} & \rev{0.633 (0.069)}\\
 & iPEB & \rev{0.984 (0.005)} & \rev{0.923 (0.041)} & \rev{0.621 (0.066)}\\
\hline
\multicolumn{5}{l}{\emph{S2: differential timing \rev{(early vs.\ late marker)}}}\\
 & PEB (panel) & \rev{0.973 (0.015)} & \rev{0.844 (0.117)} & \rev{0.252 (0.347)}\\
 & iPEB & \rev{0.995 (0.002)} & \rev{0.998 (0.005)} & \rev{1.194 (0.154)}\\
\hline
\multicolumn{5}{l}{\emph{S3: longitudinal drift + serial correlation}}\\
 & \rev{PEB (intercept, i.i.d.)} & \rev{0.880 (0.023)} & \rev{0.529 (0.065)} & \rev{0.548 (0.154)}\\
 & \rev{iPEB (intercept, i.i.d.)} & \rev{0.877 (0.024)} & \rev{0.518 (0.068)} & \rev{0.563 (0.144)}\\
 & \rev{iPEB (slope + OU)} & \rev{0.974 (0.007)} & \rev{0.846 (0.049)} & \rev{0.632 (0.000)}\\
\hline
\multicolumn{5}{l}{\snew{\emph{S4: irregular visit spacing (S3 mechanism, uneven gaps)}}}\\
 & \snew{PEB (gap-ignorant)} & \snew{0.785 (0.031)} & \snew{0.275 (0.061)} & \snew{0.531 (0.384)}\\
 & \snew{iPEB (gap-aware)} & \snew{0.978 (0.007)} & \snew{0.892 (0.044)} & \snew{0.351 (0.123)}\\
\hline
\end{tabular}
\end{table}

\begin{table}[htbp]\centering
\caption{\rev{Patient characteristics of the PLCO lung analysis cohort (screening-eligible ever-smokers, $\ge 10$ pack-years): overall, by case/control status, and by the fixed center-based train/test split (train centers 2,4,5,6,8,9; test centers 1,3,10,11). The training centers hold 225 cases and 1{,}177 controls, and the test centers 99 cases and 497 controls. Cases are older and heavier smokers and contribute fewer draws, since follow-up ends at diagnosis.}}
\label{tab:char}
\begin{tabular}{lccccc}
\hline
 & \rev{Overall} & \rev{Cases} & \rev{Controls} & \rev{Train} & \rev{Test}\\
\hline
\rev{Subjects} & \rev{1998} & \rev{324} & \rev{1674} & \rev{1402} & \rev{596}\\
\rev{Blood draws (total)} & \rev{7807} & \rev{956} & \rev{6851} & \rev{5403} & \rev{2404}\\
\rev{Draws/subject (median)} & \rev{4} & \rev{3} & \rev{4} & \rev{4} & \rev{4}\\
\rev{Age, mean (SD)} & \rev{62.6 (5.4)} & \rev{64.6 (5.2)} & \rev{62.2 (5.3)} & \rev{62.5 (5.4)} & \rev{62.6 (5.2)}\\
\rev{Pack-years, mean (SD)} & \rev{43.2 (29.1)} & \rev{61.8 (33.5)} & \rev{39.6 (26.7)} & \rev{43.4 (28.9)} & \rev{42.6 (29.3)}\\
\hline
\end{tabular}
\end{table}

\begin{table}[htbp]\centering
\caption{\rev{PLCO lung cohort, held-out TEST centers (single fixed center split; 99 case / 497 control subjects). Three like-for-like comparisons at a common panel size: four-marker (frozen \amk{4MP}), six-marker (frozen logistic), and selected-four (both sides reduced to four markers---logistic selection vs.\ iPEB objective-driven backward selection). Sensitivity is for iPEB under the sensitivity objective; median lead time (years) for iPEB under the lead-time objective; the AUC column gives test-set discrimination under that same objective---iPEB's AUC exceeds the panel's in every comparison under the sensitivity objective, and matches or exceeds it under the lead-time objective. Specificity is \emph{per visit} (the operating threshold places the stated fraction of control \emph{visits} below it); sensitivity and lead time are \emph{per patient} over the whole pre-diagnostic trajectory (a case is detected if any visit crosses; lead time is measured from the earliest crossing). Single split, so point estimates. Representative six-marker logistic coefficients (shown for interpretation, from a full-data fit on the log-normalized markers; in the analysis the panel is refit on the training centers only): intercept $-13.71$, CA125 $0.554$, CEA $1.291$, CYFRA~21-1 $0.256$, pro-SFTPB $1.494$, OPN $-0.205$, HE4 $0.189$.}}
\label{tab:lung}
\begin{tabular}{llcccc}
\hline
Comparison & Method & Spec $0.60$ & Spec $0.95$ & Spec $0.99$ & AUC\\
\hline
\multicolumn{6}{l}{\emph{\rev{Sensitivity (sensitivity objective)}}}\\
\rev{Four-marker}   & \rev{PEB (\amk{4MP})}       & \rev{0.949} & \rev{0.434} & \rev{0.182} & \rev{0.875}\\
                    & \rev{iPEB}                 & \rev{0.939} & \rev{0.485} & \rev{0.242} & \rev{0.877}\\
\rev{Six-marker}    & \rev{PEB (logistic)}       & \rev{0.960} & \rev{0.414} & \rev{0.162} & \rev{0.869}\\
                    & \rev{iPEB}                 & \rev{0.939} & \rev{0.485} & \rev{0.242} & \rev{0.876}\\
\rev{Selected-four} & \rev{PEB (logistic sel.)}  & \rev{0.960} & \rev{0.455} & \rev{0.172} & \rev{0.871}\\
                    & \rev{iPEB (backward sel.)} & \rev{0.919} & \rev{0.515} & \rev{0.192} & \rev{0.873}\\
\hline
\multicolumn{6}{l}{\emph{\rev{Median lead time, yr (lead-time objective)}}}\\
\rev{Four-marker}   & \rev{PEB (\amk{4MP})}       & \rev{2.60} & \rev{2.30} & \rev{1.09} & \rev{0.875}\\
                    & \rev{iPEB}                 & \rev{2.49} & \rev{2.01} & \rev{1.21} & \rev{0.875}\\
\rev{Six-marker}    & \rev{PEB (logistic)}       & \rev{2.59} & \rev{2.30} & \rev{1.06} & \rev{0.869}\\
                    & \rev{iPEB}                 & \rev{2.49} & \rev{1.93} & \rev{1.21} & \rev{0.876}\\
\rev{Selected-four} & \rev{PEB (logistic sel.)}  & \rev{2.59} & \rev{2.11} & \rev{1.13} & \rev{0.871}\\
                    & \rev{iPEB (backward sel.)} & \rev{2.55} & \rev{1.90} & \rev{1.13} & \rev{0.887}\\
\hline
\end{tabular}
\end{table}


\begin{figure}[p]\centering
\begin{minipage}[t]{0.49\linewidth}\centering
\includegraphics[width=\linewidth]{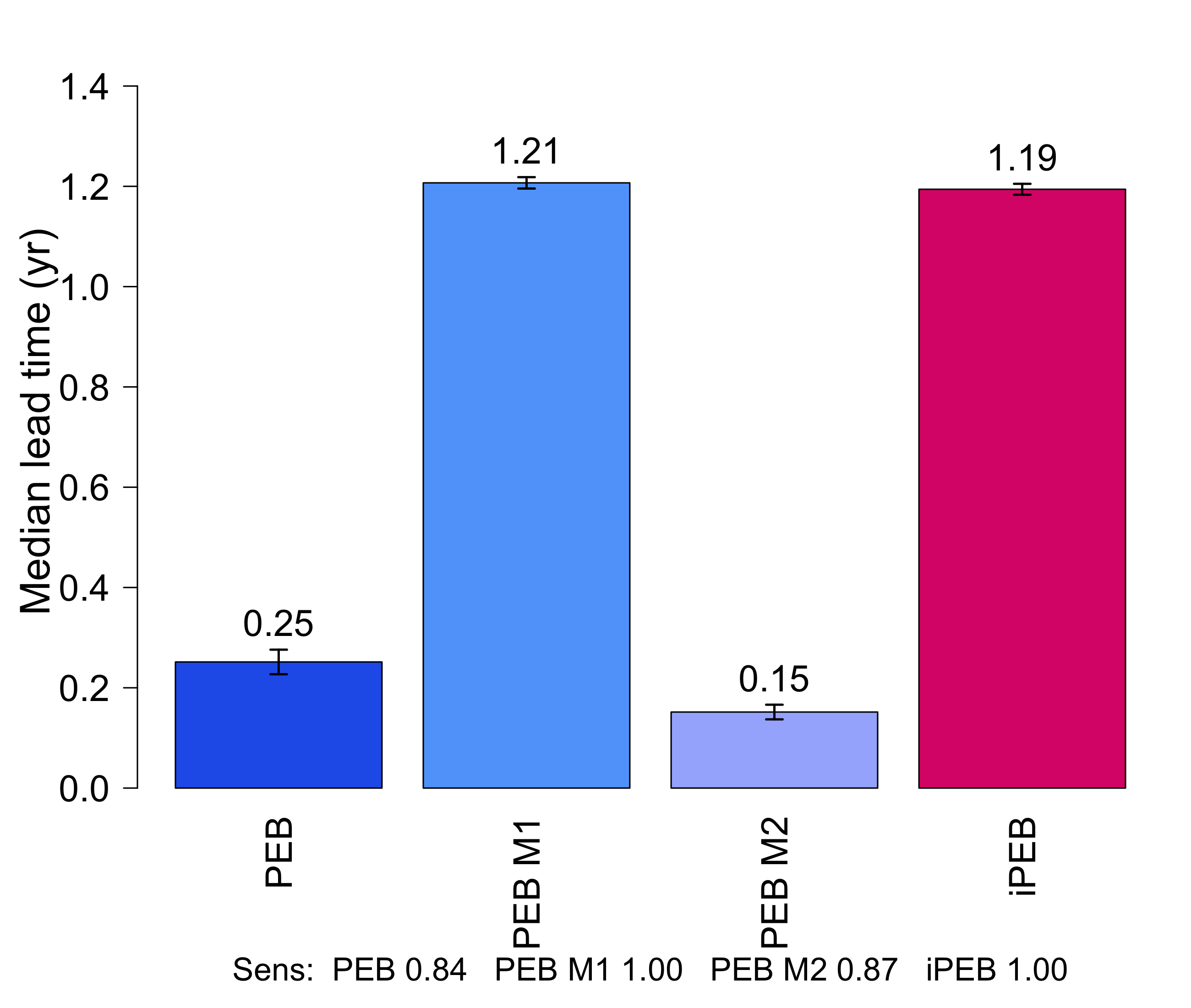}\\[4pt]
{\footnotesize \textbf{(a)} \rev{Median lead time by method (each method's sensitivity annotated below).}}
\end{minipage}\hfill
\begin{minipage}[t]{0.49\linewidth}\centering
\includegraphics[width=\linewidth]{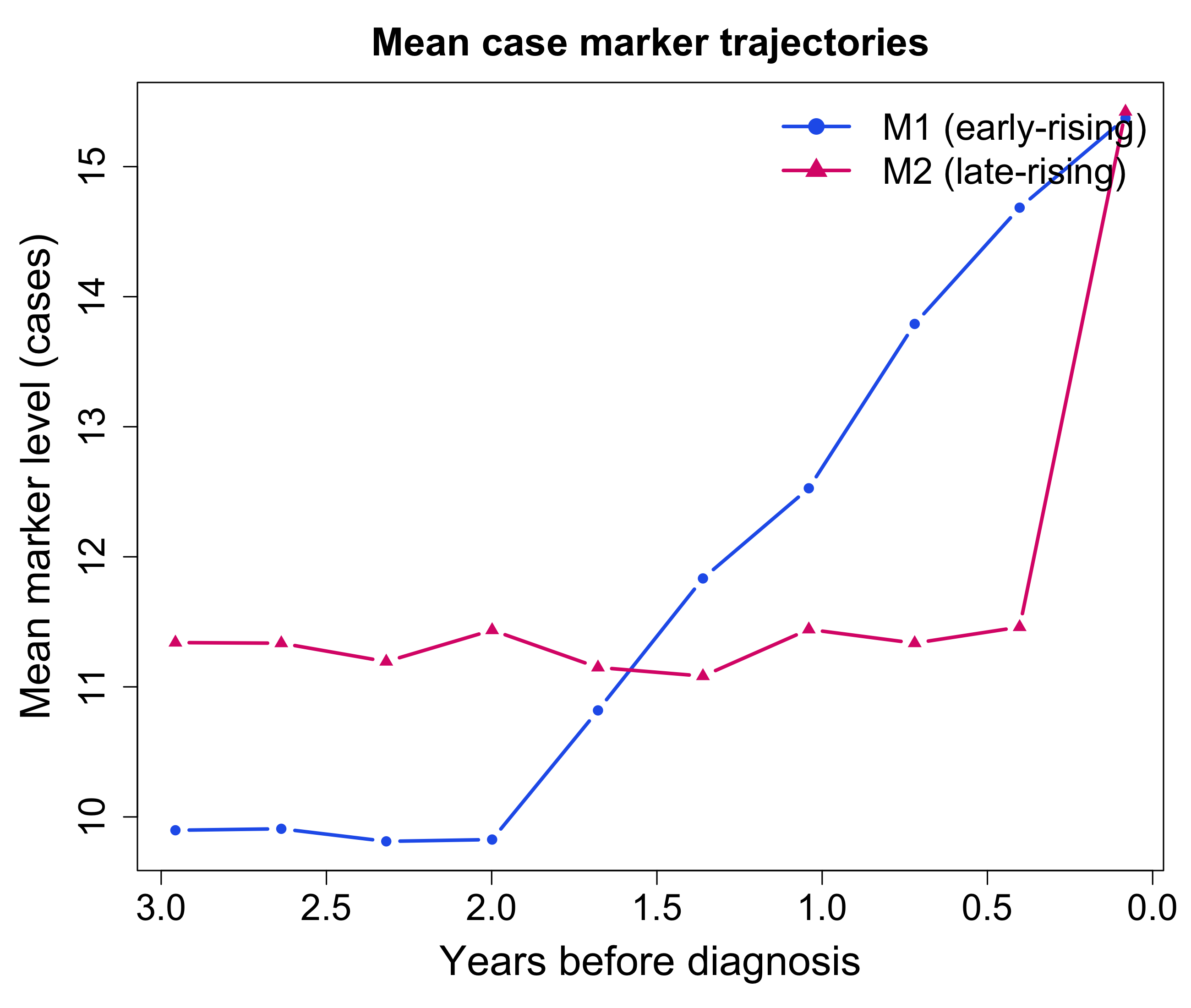}\\[4pt]
{\footnotesize \textbf{(b)} \rev{Mean case marker trajectories: M1 rises early, M2 rises late before diagnosis.}}
\end{minipage}
\caption{\rev{Mechanism Simulation~2 (differential marker timing), 200 cohorts at $95\%$ specificity. \textbf{(a)}~Median lead time for the fixed panel, its single-marker variants, and iPEB under the lead-time objective: iPEB attains the longest median lead time (about $1.19$ versus $0.25$ years for the fixed panel) while lifting sensitivity to nearly one ($1.00$ versus $0.84$). \textbf{(b)}~The underlying case means show why: marker~M1 rises early and M2 rises late, so a single fixed panel must favor either early detection (M1: long lead, lower sensitivity) or late detection (M2: shorter lead, higher sensitivity), whereas iPEB combines both. \amk{Error bars in (a): $\pm1$ Monte Carlo standard error of the mean over the 200 cohorts.}}}
\label{fig:s2}

\medskip
\begin{minipage}{0.92\linewidth}\footnotesize\raggedright
\textit{Alt text:} (a) A bar chart of median lead time by method with the iPEB bar tallest and the fixed-panel bar shortest\amk{, each bar carrying a small error whisker}; (b) two case-mean curves, one rising early and one rising late before diagnosis.
\end{minipage}
\end{figure}

\begin{figure}[p]\centering
\begin{minipage}[t]{0.49\linewidth}\centering
\includegraphics[width=\linewidth]{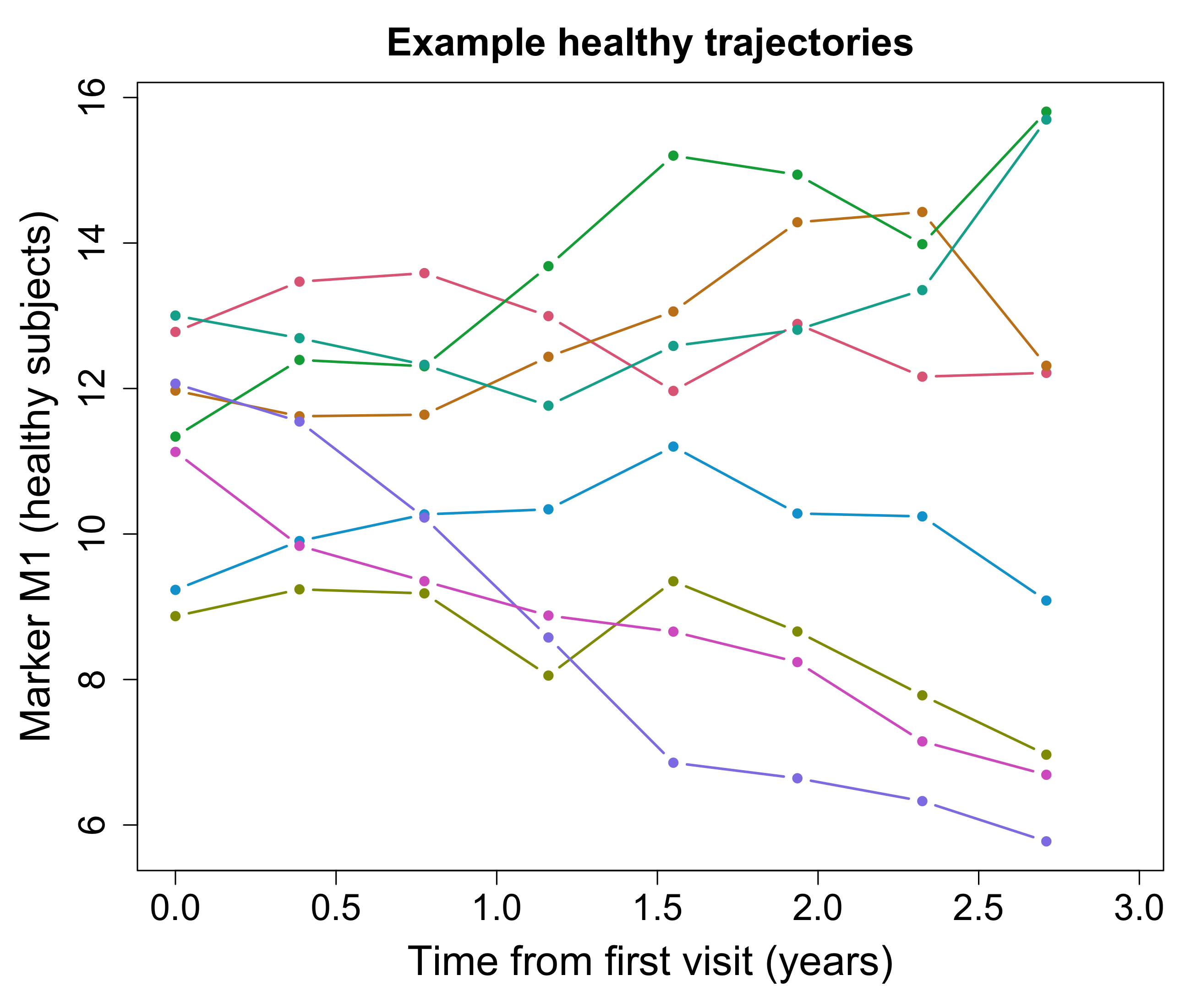}\\[4pt]
{\footnotesize \textbf{(a)} \rev{Example healthy trajectories: marker levels drift over time under a random slope with AR(1)/OU noise.}}
\end{minipage}\hfill
\begin{minipage}[t]{0.49\linewidth}\centering
\includegraphics[width=\linewidth]{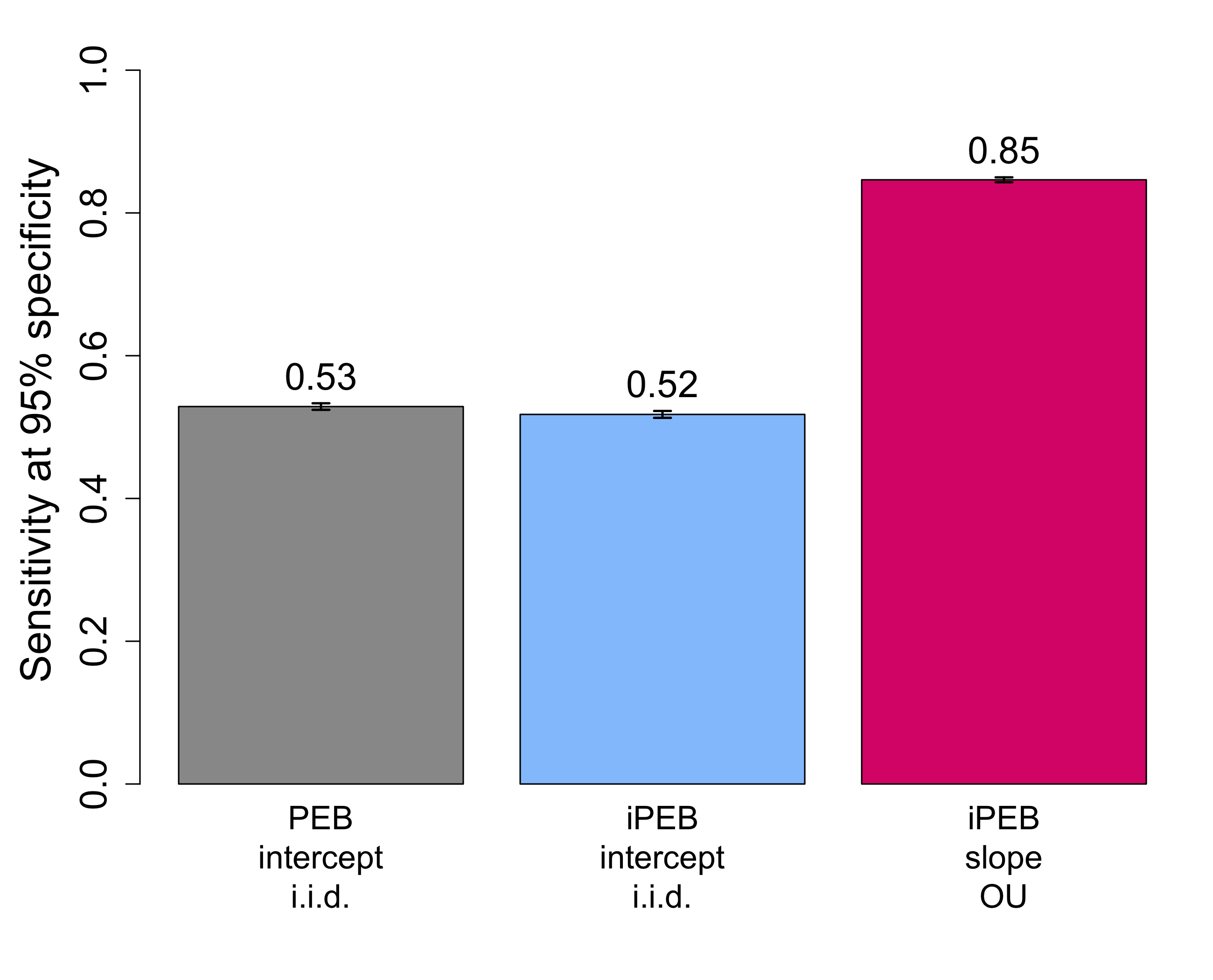}\\[4pt]
{\footnotesize \textbf{(b)} \rev{Sensitivity by method: modeling the trend recovers the sensitivity a naive layer forfeits.}}
\end{minipage}
\caption{\rev{Mechanism Simulation~3 (longitudinal drift with serial correlation), 200 cohorts at $95\%$ specificity. \textbf{(a)}~Healthy subjects carry a genuine random slope with AR(1)/OU residual autocorrelation, so their marker levels drift over time. \textbf{(b)}~An intercept-only, i.i.d.\ empirical Bayes layer mistakes this drift for signal and forfeits sensitivity (PEB intercept/i.i.d.\ $0.53$; iPEB intercept/i.i.d.\ $0.52$), whereas the time-gap--aware layer with a subject-specific slope and AR(1)/OU innovation variance removes the healthy drift and restores sensitivity to $0.85$. Because the marker weighting is held fixed across the three bars, the gain is attributable to the time-gap--aware layer itself. \amk{Error bars in (b): $\pm1$ Monte Carlo standard error of the mean over the 200 cohorts.}}}
\label{fig:s3}

\medskip
\begin{minipage}{0.92\linewidth}\footnotesize\raggedright
\textit{Alt text:} (a) Several healthy marker trajectories that drift over time; (b) a three-bar chart of sensitivity in which the slope-plus-OU iPEB bar is much taller than the naive PEB and i.i.d.\ iPEB bars\amk{, each bar carrying a small error whisker}.
\end{minipage}
\end{figure}

\begin{figure}[htbp]\centering
\includegraphics[width=0.90\linewidth]{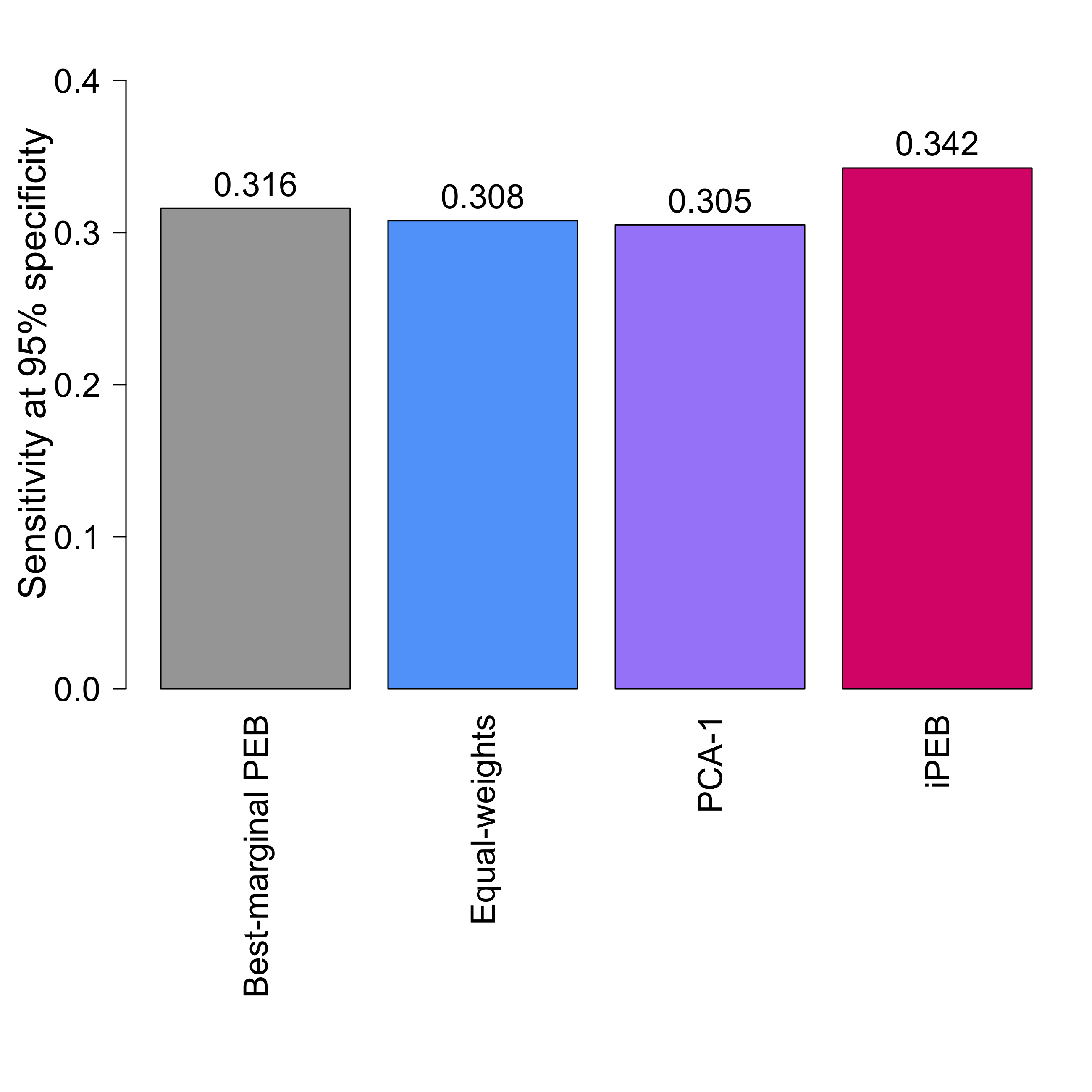}
\caption{\rev{\amk{Factorial study}: mean per-patient sensitivity at $95\%$ per-visit specificity (whole-trajectory detection) across the 48 design cells, comparing iPEB (sensitivity objective) with the best single marginal PEB marker, an equal-weight composite, and PCA-1. iPEB's mean sensitivity across cells is \amk{$0.342$}, versus $0.316$ best-marginal, $0.308$ equal-weight, and $0.305$ PCA-1. iPEB wins \amk{$45/48$} cells versus equal-weight, \amk{$45/48$ versus PCA-1}, and $39/48$ versus the stronger best-marginal PEB marker. All $9$ cells it does not win against the best marginal marker are log-normal designs, where one dominant heavy-tailed marker carries the extreme values; but even under this multiplicative skew iPEB wins $3$ of the $12$ log-normal cells, all at higher dimension. iPEB separates most under marker correlation and higher dimension ($p=10$), where objective-driven weighting has the most to exploit.}}
\label{fig:ext}

\medskip
\begin{minipage}{0.92\linewidth}\footnotesize\raggedright
\textit{Alt text:} A grouped bar chart of mean sensitivity across four methods; the iPEB bar is the tallest, followed by the best-marginal PEB, equal-weight, and PCA-1 bars.
\end{minipage}
\end{figure}


\phantomsection\label{lastpage}

\end{document}


\maketitle

These Web-based Supplementary Materials collect the proofs of the propositions
stated in Section~3 of the main paper (\hyperref[app:proofs]{Web Appendix~A}), the state-space
recursions underlying the time-gap--aware layer (\hyperref[app:kalman]{Web Appendix~B}), the
weight-optimization derivation including a robust multi-window extension
(\hyperref[app:weights]{Web Appendix~C}), full simulation-design details (\hyperref[app:simdesign]{Web Appendix~D}), the
complete per-cell results of the factorial study (\hyperref[tab:ext-p5]{Web Tables~1--2}), the full PLCO lung results across all objectives on both the test and training centers, the selected-four comparison, and the subject-level detection counts (\hyperref[tab:supp-lung-full]{Web Tables~3--6}), the mean ROC for mechanism Simulation~1 (\hyperref[webfig:s1]{Web Figure~1}), and the sensitivity comparison for mechanism Simulation~4 (\hyperref[webfig:s4]{Web Figure~2}). Equation,
proposition, and notation conventions follow the main text.

\section{Proofs of the propositions}
\label{app:proofs}
Throughout, $\nu_{ijt}$ is the one-step innovation from model~(2) of the main
text, $S_{ijt}$ its gap-dependent prediction variance (equation~(4)),
$R_{ijt}=\nu_{ijt}/\sqrt{S_{ijt}}$, $\Sigma_R=\mathrm{Cov}_{\mathcal H}(R)$, and
$c(w,\alpha)=z_{\alpha}\sqrt{w^\top\Sigma_R w}$. $\mathcal F_{t-1}$ denotes
the history up to the previous visit.

\begin{proposition}[Calibration]
Under the healthy model, for any $i,t$ and gap $\Delta_{it}>0$,
$\mathbb E_{\mathcal H}[R_{ijt}\mid\mathcal F_{t-1}]=0$,
$\mathrm{Var}_{\mathcal H}(R_{ijt}\mid\mathcal F_{t-1})=1$, and
$R_{ijt}\sim\mathcal N(0,1)$. Hence
$\Pr_{\mathcal H}\{S_i(t;w)\le c(w,\alpha)\}=\alpha$ for every $t$ and $\Delta_{it}$; the per-visit specificity equals $\alpha$.
\end{proposition}
\begin{proof}
The linear-Gaussian model~(2) admits the Kalman prediction-error decomposition:
conditional on $\mathcal F_{t-1}$, the one-step predictor $\widehat m_{ijt}$ is the
conditional mean and the innovation $\nu_{ijt}=Y_{ijt}-\widehat m_{ijt}$ satisfies
$\nu_{ijt}\mid\mathcal F_{t-1}\sim\mathcal N(0,S_{ijt})$, where $S_{ijt}$ is the
one-step prediction variance from the filter (\hyperref[app:kalman]{Web Appendix~B}), which reduces to the
gap function of equation~(4) in the intercept-only regime. Standardizing, $R_{ijt}=\nu_{ijt}/\sqrt{S_{ijt}}\mid
\mathcal F_{t-1}\sim\mathcal N(0,1)$, gives the stated conditional mean and
variance; marginalizing preserves $\mathcal N(0,1)$. For the composite,
$S_i(t;w)=w^\top R_{it}$ with $R_{it}\mid\mathcal F_{t-1}\sim\mathcal N(0,\Sigma_R)$ (taking the
healthy cross-marker innovation correlation to be gap-stable, so the pooled $\Sigma_R$ serves as
the per-visit correlation),
so $S_i(t;w)\mid\mathcal F_{t-1}\sim\mathcal N(0,w^\top\Sigma_R w)$ and
$\Pr_{\mathcal H}\{S_i(t;w)\le z_{\alpha}\sqrt{w^\top\Sigma_R w}\}=\Phi(z_{\alpha})=\alpha$,
independently of $t$ and $\Delta_{it}$.
\end{proof}

\begin{proposition}[Minimum-variance prediction]
The \rev{time-gap--aware} empirical Bayes (Kalman/BLUP) one-step predictor $\widehat m_{ijt}$
minimizes prediction mean-squared error among linear unbiased predictors using the
same history; hence the numerator of $R_{ijt}$ has minimum variance, and the layer
achieves no larger prediction MSE than mean-type baselines.
\end{proposition}
\begin{proof}
Given the correct second-order structure---the random-effects covariance $D_j$ and
the AR(1)/OU residual covariance with the observed gaps $\Delta_{it}$---the
generalized least squares / best linear unbiased predictor is, by the Gauss--Markov
theorem for correlated errors, the minimum-variance linear unbiased predictor of
$Y_{ijt}$ from $\mathcal F_{t-1}$. The Kalman filter computes exactly this BLUP
recursively (its gain is the GLS projection of the new observation onto the history),
so $\widehat m_{ijt}$ attains the minimum achievable prediction MSE. Mean-type
predictors---the subject mean, last-observation-carried-forward, and the i.i.d.\
empirical Bayes predictor that ignores $\Delta_{it}$---are (marginally) linear predictors that
use a misspecified (or degenerate) covariance, so in the model-based comparison
their prediction MSE is no smaller; strict inequality holds whenever the residual correlation is
nonzero and the gaps vary. Consequently $\mathrm{Var}(\nu_{ijt}\mid\mathcal F_{t-1})=
S_{ijt}$ is minimal among such predictors.
\end{proof}

\begin{proposition}[Power advantage over gap-ignorant rules]
For a case alternative shifting the mean by $\delta_j>0$ on a marker subset within the
pre-diagnostic window, and fixed per-visit specificity $\alpha$ for every gap, the
Neyman--Pearson most powerful test thresholds a linear form of $R_{ijt}$ (univariate
$R_{ijt}$; multivariate $w^\top R_{it}$ with $w\propto\Sigma_R^{-1}\mu_\Delta$). Any
statistic that does not standardize by $\sqrt{S_{ijt}}$ cannot be most powerful for all
$\Delta_{it}$ simultaneously; hence the time-gap--aware rule has weakly higher
true-positive rate at fixed specificity, strictly so whenever $\Delta_{it}$ varies.
\end{proposition}
\begin{proof}
Fix a gap $\Delta_{it}$ and work on the standardized scale $R_{it}=D_t^{-1}\nu_{it}$.
Under the healthy null $R_{it}\mid\mathcal F_{t-1}\sim\mathcal N(0,\Sigma_R)$ and under
the case alternative $R_{it}\mid\mathcal F_{t-1}\sim\mathcal N(\mu_\Delta,\Sigma_R)$,
where $\mu_\Delta$ is the case mean shift of $R_{it}$ (equivalently $m_\Delta=D_t\mu_\Delta$
on the raw innovation scale). For two Gaussians with common covariance the likelihood
ratio is monotone in $\mu_\Delta^\top\Sigma_R^{-1}R_{it}$ (multivariate) or in $R_{ijt}$
(univariate), so by the Neyman--Pearson lemma the most powerful test at specificity
$\alpha$ rejects for large
$w^\top R_{it}$ with $w\propto\Sigma_R^{-1}\mu_\Delta$, the maximizer of the
noncentrality $\Lambda(w)=w^\top\mu_\Delta/\sqrt{w^\top\Sigma_R w}$ (a generalized
Rayleigh quotient; see \hyperref[app:weights]{Web Appendix~C}). A statistic $T$ that omits the factor
$1/\sqrt{S_{ijt}}$ uses gap-dependent effective thresholds: to hold specificity
$\alpha$ for every gap it must vary its cut with $\Delta_{it}$, but then it is a
different (suboptimal) linear function of $\nu_{it}$ for at least one gap and cannot
coincide with the most powerful test at all gaps simultaneously. Averaging the
per-gap power over the empirical gap distribution, the standardized rule therefore has
true-positive rate at least as large, with strict inequality whenever the gap
distribution is nondegenerate.
\end{proof}

\paragraph{Diagnostics and robustness.} The propositions are model-based
(Gaussianity, correct mean/variance structure). In practice we check the healthy
innovations satisfy $R_{ijt}\approx\mathcal N(0,1)$ via normal quantile--quantile
plots, verify the residual autocorrelation is captured by the fitted $\phi_j$
(or $\lambda_j$) through the innovation autocorrelation function, and confirm that the
realized specificity is stable across calendar time and assay batch by
stratified-specificity plots. Under heavy tails or change-points calibration holds
approximately and the power advantage attenuates but does not reverse.

\section{State-space (Kalman) details}
\label{app:kalman}
Write each marker as a linear Gaussian state-space model with state
$\theta_{ijt}=(b^{(0)}_{ij,t},b^{(1)}_{ij,t})^\top$, transition
$T_j(\Delta_{it})=\left[\begin{smallmatrix}1&\Delta_{it}\\0&1\end{smallmatrix}\right]$
(local linear trend) or $I_2$ (static random effects), and observation
$Y_{ijt}=\mu_j+x_{it}^\top\beta_j+Z_{it}\theta_{ijt}+\varepsilon_{ijt}$ with observation
row $Z_{it}=[1\ h_{it}]$ for the static ($I_2$) random-effects formulation---which reproduces
the main-text model $b^{(0)}_{ij}+b^{(1)}_{ij}h_{it}$ of equation~(2)---or $Z_{it}=[1\ 0]$
when the local-linear-trend transition is used (the level state then already integrates the
slope), and AR(1)/OU observation noise. The filter recursion is: predict
$a_{ijt}=T_j m_{ij,t-1}$ and $P^-=T_jP_{ij,t-1}T_j^\top+R_jQ_jR_j^\top$; form the
innovation $\nu_{ijt}=Y_{ijt}-\mu_j-x_{it}^\top\beta_j-Z_{it}a_{ijt}$ with variance
$S_{ijt}=Z_{it}P^-Z_{it}^\top+\sigma_j^2\{1-\rho_j^2(\Delta_{it})\}$; update via the
Kalman gain $K=P^-Z_{it}^\top S_{ijt}^{-1}$; and standardize
$R_{ijt}=\nu_{ijt}/\sqrt{S_{ijt}}$. Hyperparameters are estimated by REML/maximum
likelihood through the prediction-error decomposition; set $Q_j=0$ (static random
effects) when most subjects have $\le2$ visits, use diffuse initialization
$P_{j0}=\kappa I$, and compute $\rho_j(\Delta)$ and $S_{ijt}$ from the actual gaps
$\Delta_{it}$. Implementations use standard state-space routines; we validate that the
healthy $R_{ijt}$ are approximately $\mathcal N(0,1)$ and that $\Sigma_R$ is stable
across calendar time and batch.

\section{Weight optimization: derivation and extensions}
\label{app:weights}
On the innovation scale a healthy visit has $R_{it}\sim\mathcal N(0,\Sigma_R)$ and a
case in window has $R_{it}\sim\mathcal N(\mu_\Delta,\Sigma_R)$. For the composite
$S_i(t;w)=w^\top R_{it}$ the true-positive rate at specificity $\alpha$ is
$1-\Phi\{z_{\alpha}-\Lambda(w)\}$ with noncentrality
$\Lambda(w)=w^\top\mu_\Delta/\sqrt{w^\top\Sigma_R w}$; maximizing power is thus
maximizing $\Lambda(w)$. Differentiating the generalized Rayleigh quotient, the
stationarity condition $\Sigma_R w\,(w^\top\mu_\Delta)=\mu_\Delta\,(w^\top\Sigma_R w)$
gives $w\propto\Sigma_R^{-1}\mu_\Delta$; normalizing so $w^{\star\top}\Sigma_R w^\star=1$
yields the closed form (equation~(7) of the main text), and projecting onto the
nonnegative orthant if required. This is the population optimum; in finite samples
$\mu_\Delta$ and $\Sigma_R$ are replaced by \amk{their sample counterparts from the tuning
split, with $\Sigma_R$ regularized by a small ridge on the diagonal for numerical stability}.

\paragraph{\rev{Lead-time objective and the sensitivity floor.}} \rev{The three objectives
share the loss~(6) and differ only in the fixed profile $(\lambda_1,\lambda_2,\lambda_3)$.
Because the threshold~(5) is recalibrated to $\alpha$ for every candidate $w$, the specificity
term $\lambda_1\{\alpha-\operatorname{Spec}(w)\}_{+}$ is essentially inactive; under the
lead-time profile $(50,0,0.7)$ the operative trade-off is therefore between the retained
sensitivity term $-\operatorname{Sens}_W(w)$, which carries unit weight, and the lead-time
reward $-\lambda_3\operatorname{LT}(w)$, with $\lambda_3=0.7$ and $\operatorname{LT}$ entered on
the two-year normalized scale $\operatorname{LT}/\operatorname{LT}_{\mathrm{ref}}$,
$\operatorname{LT}_{\mathrm{ref}}=2$. Since the unit-weight sensitivity term outweighs a
fractional-year lead increment, the optimizer lengthens lead time only among weightings that
keep sensitivity high; the tie-breakers (higher $\operatorname{Sens}_W$ first, then lower
referral, then sparser weights) further rule out the degenerate solution that would flag a
single late-rising case to inflate the median lead. Lead time is thus optimized subject to a
sensitivity floor rather than in isolation.}

\paragraph{Robust multi-window / multi-subtype extension.} When detection must be
robust across several windows or disease subtypes $\{\mu_{\Delta,k}\}_{k=1}^K$,
maximize the worst-case noncentrality
\[
\max_{w}\ \min_{k}\ \frac{w^\top\mu_{\Delta,k}}{\sqrt{w^\top\Sigma_R w}}.
\]
Fixing the scale $w^\top\Sigma_R w=1$ and introducing an epigraph variable $s$, this is
$\max_{w,s}\,s$ subject to $w^\top\mu_{\Delta,k}\ge s$ for all $k$ and
$\|\Sigma_R^{1/2}w\|_2\le1$, a second-order cone program---hence convex with a global
optimum. Ellipsoidal uncertainty $\mu_{\Delta,k}\in\{\bar\mu_k+E_k u:\|u\|_2\le1\}$
replaces each linear constraint by $w^\top\bar\mu_k-\|E_k^\top w\|_2\ge s$, preserving
the second-order cone structure. When empirical clinical metrics are optimized
directly the objective is non-convex and we use scalarization of the loss
(equation~(6)) with the stated tie-breakers, or a population-based multi-objective
search (NSGA-II) for a Pareto frontier \citep{Deb2002}.

\paragraph{Complexity.} With $n$ subjects, $p$ markers, mean visit count $\bar T$, and
$G$ evaluated weight vectors: fitting the layer is $O(p\,n\bar T)$; estimating
$\Sigma_R$ is $O(p^2 n\bar T)$; per-$w$ scoring is $O(p\,n\bar T)$ and the variance
$O(p^2 n\bar T)$ (or $O(p\,n\bar T)$ with a diagonal/shrunk $\Sigma_R$). Caching $R$,
using a shrunk $\Sigma_R$, and precomputing a Cholesky factor for
$v=\|\,\mathrm{chol}(\Sigma_R)\,w\|_2^2$ keep the search efficient; evaluations
parallelize over $w$ and subjects. \snew{In practice these costs are small: a single end-to-end fit of the six-marker cohort under the sensitivity objective (weight optimization plus layer fitting and scoring) completes in about $7$ seconds on a 2023 laptop (Apple M2~Pro, 16~GB RAM), so a full sweep over the three objectives and three comparisons runs in about a minute and no special tuning is required.}

\section{Simulation design details}
\label{app:simdesign}
Healthy subjects follow $Y_{ijt}=\mu_j+b_{ij}+\varepsilon_{ijt}$ with cross-marker
correlation applied to $(\varepsilon_{i1t},\dots,\varepsilon_{ipt})$; cases share the
baseline distribution then drift over the final $W_0$ months,
$Y_{ijt}=\mu_j+b_{ij}+\varepsilon_{ijt}+\delta_j(1-\tfrac{T_i-t}{W_0})_+
\mathbf 1\{t\in[T_i-W_0,T_i)\}$ for $j\in\mathcal D$. Thus cases and controls share an
identical baseline and diverge only within the pre-diagnostic window, so a marker
carries signal only when observed late enough---the defining feature of a screening
problem---and the drift set $\mathcal D$ with unequal effects $\delta_j$ makes the
informative subset both sparse and heterogeneous.

\paragraph{Design rationale.} Each factorial axis isolates a regime that separates
objective-driven weighting from fixed composites. \emph{Dimension} $p\in\{5,10\}$
increases the number of markers a combiner must balance; equal-weight and PCA-1 dilute
signal with noise as $p$ grows, whereas the discriminant direction
$\Sigma_R^{-1}\mu_\Delta$ does not. \emph{Correlation} (low, high, two-block) controls
redundancy among markers; correlated noise is exactly where whitening by
$\Sigma_R^{-1}$ helps most and where naive averaging is least efficient, and the
two-block structure mimics pathway-grouped assays. \emph{Sparsity} $|\mathcal D|/p$
sets the fraction of markers that actually carry disease signal, testing robustness to
diluting the composite with uninformative markers. \emph{Misspecification} stresses the
Gaussian working model: $t_5$ adds heavy tails, AR(1) adds residual autocorrelation,
and log-normal adds multiplicative skew---the last being the adversarial case in which
one dominant marginal marker can beat any composite, \rev{which is why every one of the
$9$ cells where the best single marker exceeds iPEB is a log-normal cell; even there, once
the dimension is high enough for signal to be spread across markers, iPEB still wins $3$ of
the $12$ log-normal cells.} This mapping lets the per-cell
results (\hyperref[tab:ext-p5]{Web Tables~1--2}) be read directly: gains grow with dimension and correlation,
and shrink toward the single-marker boundary under log-normal skew.

The factorial varies $p\in\{5,10\}$; correlation (low $0.2^{|j-k|}$, high $0.7^{|j-k|}$,
two-block); sparsity $|\mathcal D|/p\in\{0.4,0.6\}$ with unequal effects; and
misspecification (none, $t_5$, AR(1) with $\phi=0.4$, log-normal multiplicative),
giving 48 cells $\times$ 100 cohorts. \rev{Within each cohort a random \amk{$75\%$} of subjects
(cases and controls) forms the training sample on which the iPEB weights, the comparators,
and the operating threshold are learned; all metrics are then evaluated on the held-out
\amk{$25\%$}, and averaged over the 100 cohorts.} Comparators (best marginal PEB, equal-weight,
PCA-1) use identical PEB thresholding, with the best marginal marker chosen on the
training data by training sensitivity; iPEB is evaluated under the
sensitivity objective. \snew{The four} mechanism scenarios fix specificity at $0.95$ and, within each cohort, train on a random $70\%$ of subjects and evaluate on the held-out $30\%$: S1
(heterogeneous markers, similar timing), S2 (differential timing: \rev{an early-rising marker $\sim$2 years before diagnosis versus a late-rising marker $\sim$3 months before}, panel developed on an independent cohort), S3 (healthy
random slope with AR(1) residuals to test the time-gap--aware layer)\snew{, and S4 (the S3 mechanism on an \emph{irregularly} spaced grid; see below)}. \rev{\snew{The first three}
mechanism scenarios place every subject on a common \emph{regular} visit grid, so each
isolates a single mechanism cleanly; the factorial study instead draws each subject's
visit times \emph{irregularly}, as a mixture of semiannual and annual inter-visit gaps,
so the gap-aware AR(1)/OU innovation variance is genuinely exercised in its AR(1) cells
(as it is, naturally, on the real cohort, whose draws are irregularly spaced). Because
each simulated subject contributes at least five visits, the random slope is enabled
throughout the simulations (except in the deliberately intercept-only comparator arms of
S3 and S4); the AR(1)/OU innovation is switched on only where the
generative process carries serial correlation or drift (S3, S4, and the AR(1) factorial cells),
and reduces to i.i.d.\ otherwise. Within each simulated training set a further $25\%$ of
subjects is held out as an internal validation slice used to choose the
scalar-versus-multivariate variant (feature selection is not used in the simulations);
the test subjects are never used in fitting. Detection---both sensitivity and lead time---is
scored over the whole pre-diagnostic trajectory (a case counts as detected if any pre-diagnostic
visit crosses the threshold) in \emph{both} weight optimization and evaluation, for every
scenario and objective; no detection window is imposed. A window $W$ may optionally be supplied to
restrict the objective to a fixed pre-diagnostic horizon, but it is not used here.}

\paragraph{\snew{Scenario~S4 (irregular visit spacing).}} \snew{A fourth mechanism scenario tests the
time-gap--aware layer when visits are unevenly spaced. It uses the same two-marker,
drift-plus-serial-correlation mechanism as S3, but each subject is measured on its own
\emph{irregular} schedule of eight visits whose inter-visit gaps are drawn from a mixture of
$3$-, $6$-, $12$-, and $24$-month intervals (about three months to two years), and residuals follow
a continuous-time Ornstein--Uhlenbeck process (the same serial correlation as S3), so the elapsed gap
$\Delta$ varies within and across subjects. Cases and controls share the same spacing distribution,
so visit timing carries no information about case status. A gap-ignorant PEB (intercept-only,
i.i.d.) is compared with the gap-aware iPEB (random slope with gap-scaled AR(1)/OU innovations);
both are scored over the whole pre-diagnostic trajectory at $95\%$ per-visit specificity under the
sensitivity objective, with no detection window. Unable to rescale for the varying gaps, the
gap-ignorant rule detects far fewer cases (sensitivity $0.275$ versus $0.892$; AUC $0.785$ versus
$0.978$), and the drop from its regular-grid S3 value ($0.529$) isolates the specific cost of
ignoring irregular spacing (Table~1 of the main text; \hyperref[webfig:s4]{Web Figure~2}).}

\section*{Web Figures}
\addcontentsline{toc}{section}{Web Figures}

\begin{figure}[htbp]\centering
\includegraphics[width=\linewidth]{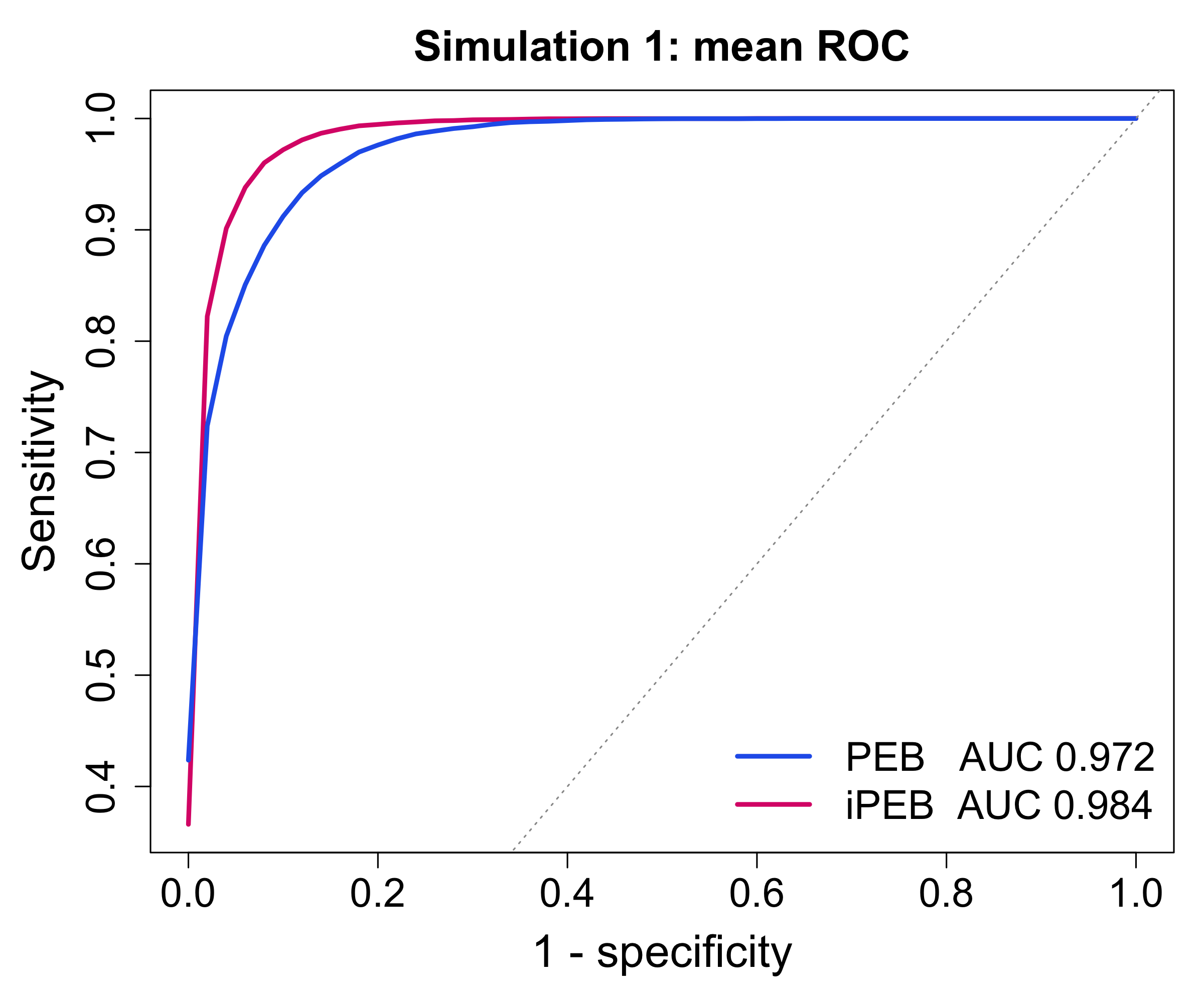}
\caption{\rev{Mechanism Simulation~1 (marker heterogeneity): mean ROC over 200 cohorts. Two case subgroups are each driven by a \emph{different} marker, so no single fixed composite aligns with both. History-adjusted PEB on the fixed composite is compared with iPEB, which reweights the same markers for the sensitivity objective. The iPEB curve lies above PEB across the whole range, raising the area under the curve from $0.972$ to $0.984$, with the largest gain in the low-false-positive (high-specificity) region relevant to screening. The corresponding summary numbers (AUC, sensitivity, lead time) are in Table~1 of the main text.}}
\label{webfig:s1}

\medskip
\begin{minipage}{0.92\linewidth}\footnotesize\raggedright
\rev{\textit{Alt text:} An ROC plot with false-positive rate on the horizontal axis and sensitivity on the vertical axis; the iPEB curve sits above the PEB curve everywhere, with labelled AUCs of $0.984$ and $0.972$.}
\end{minipage}
\end{figure}

\begin{figure}[htbp]\centering
\includegraphics[width=\linewidth]{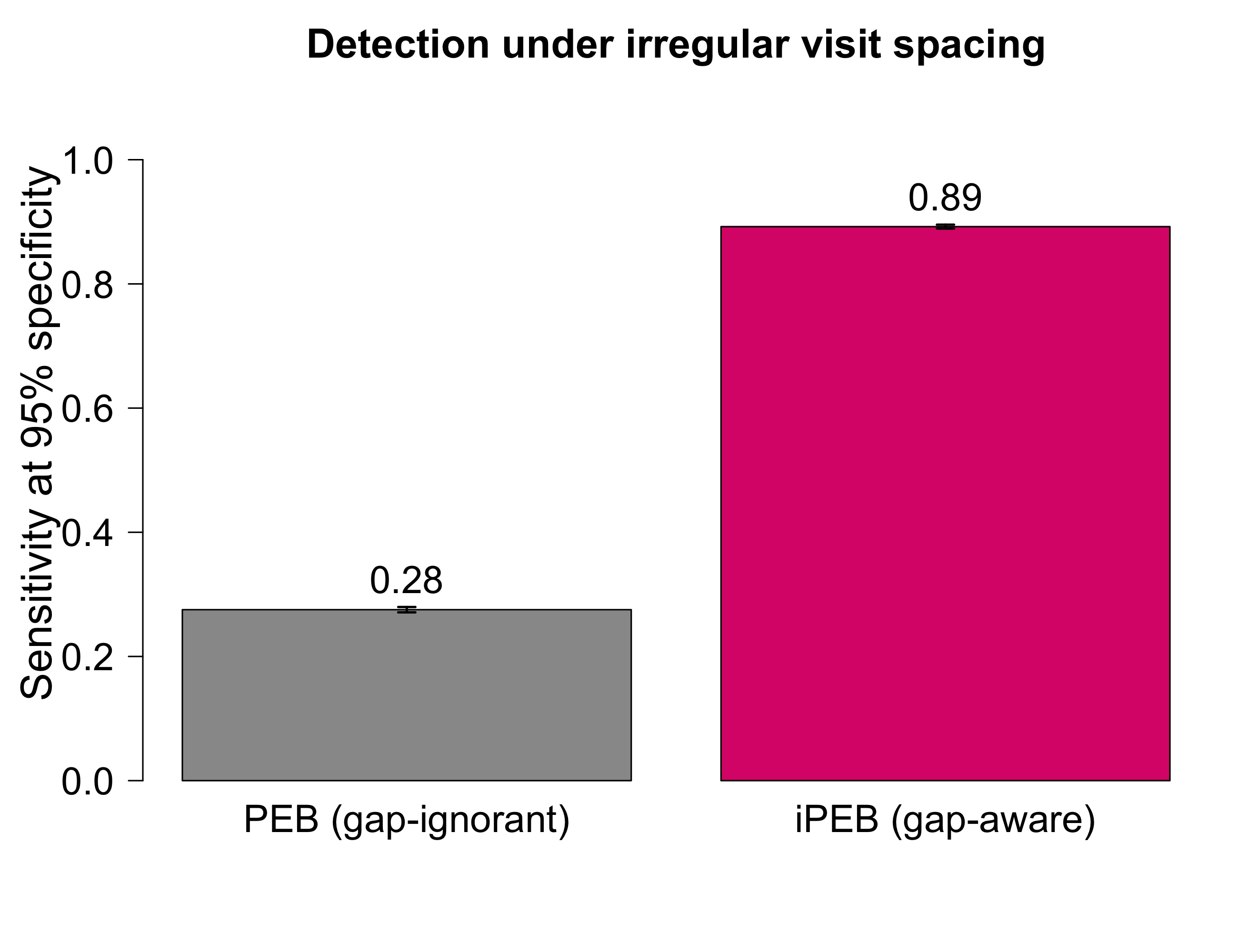}
\caption{\snew{Mechanism Simulation~4 (irregular visit spacing): per-patient sensitivity at $95\%$ per-visit specificity over 200 cohorts, under the sensitivity objective. Healthy subjects carry a random slope and continuous-time serial correlation, and each subject's visits are irregularly spaced (inter-visit gaps from about three months to two years). The gap-ignorant PEB (intercept-only, i.i.d.) cannot rescale for the varying gaps and detects few cases, whereas the gap-aware iPEB (random slope with gap-scaled AR(1)/OU innovations) detects the great majority. \amk{Error bars: $\pm1$ Monte Carlo standard error of the mean over the 200 cohorts.} Summary numbers appear in Table~1 of the main text.}}
\label{webfig:s4}

\medskip
\begin{minipage}{0.92\linewidth}\footnotesize\raggedright
\snew{\textit{Alt text:} A two-bar chart of sensitivity at $95\%$ specificity under irregular visit spacing; the gap-aware iPEB bar (about $0.89$) is much taller than the gap-ignorant PEB bar (about $0.28$)\amk{, each bar carrying a small error whisker}.}
\end{minipage}
\end{figure}

\section*{Web Tables}
\addcontentsline{toc}{section}{Web Tables}

\begin{table}[htbp]\centering\footnotesize
\caption{\rev{\amk{Factorial study}, $p=5$: per-patient sensitivity at $95\%$ per-visit specificity under whole-trajectory detection, mean\,(SD) over 100 Monte Carlo cohorts. Best method per cell in \textbf{bold}.}}
\label{tab:ext-p5}
\begin{tabular}{lll cccc}
\hline
Corr & Spar & Misspec & Best-marg.\ PEB & Equal-wt & PCA-1 & iPEB\\
\hline
\rev{low} & \rev{0.4} & \rev{none} & \rev{0.286\,(0.067)} & \rev{0.295\,(0.063)} & \rev{0.291\,(0.063)} & \rev{\textbf{0.326\,(0.067)}}\\
\rev{high} & \rev{0.4} & \rev{none} & \rev{0.299\,(0.065)} & \rev{0.294\,(0.064)} & \rev{0.291\,(0.063)} & \rev{\textbf{0.332\,(0.080)}}\\
\rev{block} & \rev{0.4} & \rev{none} & \rev{0.300\,(0.065)} & \rev{0.294\,(0.068)} & \rev{0.296\,(0.062)} & \rev{\textbf{0.324\,(0.084)}}\\
\rev{low} & \rev{0.6} & \rev{none} & \rev{0.291\,(0.073)} & \rev{0.310\,(0.065)} & \rev{0.311\,(0.060)} & \rev{\textbf{0.319\,(0.070)}}\\
\rev{high} & \rev{0.6} & \rev{none} & \rev{0.307\,(0.071)} & \rev{0.307\,(0.065)} & \rev{0.307\,(0.064)} & \rev{\textbf{0.322\,(0.077)}}\\
\rev{block} & \rev{0.6} & \rev{none} & \rev{0.303\,(0.070)} & \rev{0.308\,(0.069)} & \rev{0.305\,(0.067)} & \rev{\textbf{0.317\,(0.071)}}\\
\rev{low} & \rev{0.4} & \rev{$t_5$} & \rev{0.295\,(0.062)} & \rev{0.294\,(0.071)} & \rev{0.287\,(0.070)} & \rev{\textbf{0.301\,(0.072)}}\\
\rev{high} & \rev{0.4} & \rev{$t_5$} & \rev{0.285\,(0.067)} & \rev{0.284\,(0.066)} & \rev{0.284\,(0.068)} & \rev{\textbf{0.310\,(0.081)}}\\
\rev{block} & \rev{0.4} & \rev{$t_5$} & \rev{0.290\,(0.067)} & \rev{0.293\,(0.063)} & \rev{0.288\,(0.068)} & \rev{\textbf{0.304\,(0.074)}}\\
\rev{low} & \rev{0.6} & \rev{$t_5$} & \rev{0.302\,(0.061)} & \rev{\textbf{0.312\,(0.079)}} & \rev{0.306\,(0.074)} & \rev{0.305\,(0.066)}\\
\rev{high} & \rev{0.6} & \rev{$t_5$} & \rev{0.298\,(0.070)} & \rev{0.291\,(0.071)} & \rev{0.291\,(0.071)} & \rev{\textbf{0.314\,(0.074)}}\\
\rev{block} & \rev{0.6} & \rev{$t_5$} & \rev{0.287\,(0.070)} & \rev{0.296\,(0.070)} & \rev{0.290\,(0.075)} & \rev{\textbf{0.317\,(0.069)}}\\
\rev{low} & \rev{0.4} & \rev{AR(1)} & \rev{0.302\,(0.074)} & \rev{0.288\,(0.064)} & \rev{0.289\,(0.063)} & \rev{\textbf{0.311\,(0.067)}}\\
\rev{high} & \rev{0.4} & \rev{AR(1)} & \rev{0.307\,(0.073)} & \rev{0.286\,(0.057)} & \rev{0.287\,(0.058)} & \rev{\textbf{0.332\,(0.079)}}\\
\rev{block} & \rev{0.4} & \rev{AR(1)} & \rev{0.301\,(0.064)} & \rev{0.288\,(0.067)} & \rev{0.291\,(0.061)} & \rev{\textbf{0.332\,(0.082)}}\\
\rev{low} & \rev{0.6} & \rev{AR(1)} & \rev{0.309\,(0.075)} & \rev{0.322\,(0.071)} & \rev{0.316\,(0.068)} & \rev{\textbf{0.328\,(0.076)}}\\
\rev{high} & \rev{0.6} & \rev{AR(1)} & \rev{0.309\,(0.079)} & \rev{0.307\,(0.068)} & \rev{0.309\,(0.067)} & \rev{\textbf{0.321\,(0.087)}}\\
\rev{block} & \rev{0.6} & \rev{AR(1)} & \rev{0.310\,(0.072)} & \rev{0.313\,(0.069)} & \rev{0.304\,(0.064)} & \rev{\textbf{0.325\,(0.077)}}\\
\rev{low} & \rev{0.4} & \rev{log-normal} & \rev{\textbf{0.349\,(0.071)}} & \rev{0.313\,(0.070)} & \rev{0.309\,(0.069)} & \rev{0.336\,(0.074)}\\
\rev{high} & \rev{0.4} & \rev{log-normal} & \rev{\textbf{0.341\,(0.068)}} & \rev{0.319\,(0.071)} & \rev{0.318\,(0.068)} & \rev{0.330\,(0.073)}\\
\rev{block} & \rev{0.4} & \rev{log-normal} & \rev{\textbf{0.335\,(0.073)}} & \rev{0.311\,(0.066)} & \rev{0.315\,(0.070)} & \rev{0.324\,(0.079)}\\
\rev{low} & \rev{0.6} & \rev{log-normal} & \rev{\textbf{0.348\,(0.077)}} & \rev{0.323\,(0.071)} & \rev{0.322\,(0.070)} & \rev{0.326\,(0.073)}\\
\rev{high} & \rev{0.6} & \rev{log-normal} & \rev{\textbf{0.345\,(0.074)}} & \rev{0.326\,(0.076)} & \rev{0.325\,(0.071)} & \rev{0.325\,(0.079)}\\
\rev{block} & \rev{0.6} & \rev{log-normal} & \rev{\textbf{0.353\,(0.079)}} & \rev{0.327\,(0.070)} & \rev{0.333\,(0.074)} & \rev{0.314\,(0.072)}\\
\hline
\end{tabular}
\end{table}

\begin{table}[htbp]\centering\footnotesize
\caption{\rev{\amk{Factorial study}, $p=10$: per-patient sensitivity at $95\%$ per-visit specificity under whole-trajectory detection, mean\,(SD) over 100 Monte Carlo cohorts. Best method per cell in \textbf{bold}.}}
\label{tab:ext-p10}
\begin{tabular}{lll cccc}
\hline
Corr & Spar & Misspec & Best-marg.\ PEB & Equal-wt & PCA-1 & iPEB\\
\hline
\rev{low} & \rev{0.4} & \rev{none} & \rev{0.296\,(0.067)} & \rev{0.303\,(0.065)} & \rev{0.289\,(0.059)} & \rev{\textbf{0.339\,(0.072)}}\\
\rev{high} & \rev{0.4} & \rev{none} & \rev{0.292\,(0.072)} & \rev{0.297\,(0.067)} & \rev{0.298\,(0.070)} & \rev{\textbf{0.378\,(0.089)}}\\
\rev{block} & \rev{0.4} & \rev{none} & \rev{0.304\,(0.076)} & \rev{0.289\,(0.066)} & \rev{0.287\,(0.066)} & \rev{\textbf{0.377\,(0.084)}}\\
\rev{low} & \rev{0.6} & \rev{none} & \rev{0.312\,(0.071)} & \rev{0.321\,(0.069)} & \rev{0.310\,(0.068)} & \rev{\textbf{0.354\,(0.068)}}\\
\rev{high} & \rev{0.6} & \rev{none} & \rev{0.318\,(0.070)} & \rev{0.304\,(0.069)} & \rev{0.307\,(0.072)} & \rev{\textbf{0.372\,(0.076)}}\\
\rev{block} & \rev{0.6} & \rev{none} & \rev{0.318\,(0.072)} & \rev{0.307\,(0.063)} & \rev{0.306\,(0.063)} & \rev{\textbf{0.365\,(0.075)}}\\
\rev{low} & \rev{0.4} & \rev{$t_5$} & \rev{0.303\,(0.071)} & \rev{0.314\,(0.067)} & \rev{0.304\,(0.062)} & \rev{\textbf{0.340\,(0.077)}}\\
\rev{high} & \rev{0.4} & \rev{$t_5$} & \rev{0.306\,(0.073)} & \rev{0.295\,(0.064)} & \rev{0.291\,(0.070)} & \rev{\textbf{0.374\,(0.086)}}\\
\rev{block} & \rev{0.4} & \rev{$t_5$} & \rev{0.308\,(0.070)} & \rev{0.289\,(0.069)} & \rev{0.287\,(0.069)} & \rev{\textbf{0.374\,(0.077)}}\\
\rev{low} & \rev{0.6} & \rev{$t_5$} & \rev{0.308\,(0.073)} & \rev{0.341\,(0.069)} & \rev{0.331\,(0.066)} & \rev{\textbf{0.350\,(0.068)}}\\
\rev{high} & \rev{0.6} & \rev{$t_5$} & \rev{0.308\,(0.072)} & \rev{0.314\,(0.066)} & \rev{0.313\,(0.072)} & \rev{\textbf{0.365\,(0.076)}}\\
\rev{block} & \rev{0.6} & \rev{$t_5$} & \rev{0.312\,(0.064)} & \rev{0.307\,(0.066)} & \rev{0.307\,(0.067)} & \rev{\textbf{0.361\,(0.073)}}\\
\rev{low} & \rev{0.4} & \rev{AR(1)} & \rev{0.295\,(0.068)} & \rev{0.307\,(0.064)} & \rev{0.299\,(0.061)} & \rev{\textbf{0.341\,(0.073)}}\\
\rev{high} & \rev{0.4} & \rev{AR(1)} & \rev{0.299\,(0.068)} & \rev{0.290\,(0.067)} & \rev{0.291\,(0.070)} & \rev{\textbf{0.384\,(0.083)}}\\
\rev{block} & \rev{0.4} & \rev{AR(1)} & \rev{0.310\,(0.071)} & \rev{0.286\,(0.068)} & \rev{0.286\,(0.068)} & \rev{\textbf{0.383\,(0.085)}}\\
\rev{low} & \rev{0.6} & \rev{AR(1)} & \rev{0.318\,(0.074)} & \rev{0.335\,(0.065)} & \rev{0.324\,(0.070)} & \rev{\textbf{0.359\,(0.073)}}\\
\rev{high} & \rev{0.6} & \rev{AR(1)} & \rev{0.319\,(0.068)} & \rev{0.302\,(0.070)} & \rev{0.304\,(0.071)} & \rev{\textbf{0.380\,(0.087)}}\\
\rev{block} & \rev{0.6} & \rev{AR(1)} & \rev{0.322\,(0.070)} & \rev{0.301\,(0.062)} & \rev{0.299\,(0.063)} & \rev{\textbf{0.386\,(0.075)}}\\
\rev{low} & \rev{0.4} & \rev{log-normal} & \rev{\textbf{0.368\,(0.080)}} & \rev{0.317\,(0.065)} & \rev{0.300\,(0.068)} & \rev{0.349\,(0.077)}\\
\rev{high} & \rev{0.4} & \rev{log-normal} & \rev{0.342\,(0.083)} & \rev{0.324\,(0.063)} & \rev{0.326\,(0.068)} & \rev{\textbf{0.361\,(0.084)}}\\
\rev{block} & \rev{0.4} & \rev{log-normal} & \rev{0.343\,(0.078)} & \rev{0.318\,(0.062)} & \rev{0.317\,(0.062)} & \rev{\textbf{0.354\,(0.083)}}\\
\rev{low} & \rev{0.6} & \rev{log-normal} & \rev{0.365\,(0.077)} & \rev{0.340\,(0.067)} & \rev{0.333\,(0.073)} & \rev{\textbf{0.369\,(0.069)}}\\
\rev{high} & \rev{0.6} & \rev{log-normal} & \rev{\textbf{0.371\,(0.080)}} & \rev{0.330\,(0.069)} & \rev{0.334\,(0.071)} & \rev{0.366\,(0.079)}\\
\rev{block} & \rev{0.6} & \rev{log-normal} & \rev{\textbf{0.367\,(0.079)}} & \rev{0.336\,(0.064)} & \rev{0.337\,(0.064)} & \rev{0.359\,(0.074)}\\
\hline
\end{tabular}
\end{table}

\begin{table}[htbp]\centering\footnotesize
\caption{\rev{PLCO lung cohort, held-out test centers: complete four-marker and six-marker comparisons across all three iPEB objectives. Sensitivity (Sens) and median lead time in years (Lead) at three operating specificities; specificity is per-visit, sensitivity and lead per-patient over the whole trajectory; test-set AUC. PEB (frozen) is unaffected by the objective, which reweights only iPEB.}}
\label{tab:supp-lung-full}
\begin{tabular}{l ccc ccc c}
\hline
 & \multicolumn{3}{c}{Sensitivity (by spec.)} & \multicolumn{3}{c}{Lead time, yr (by spec.)} & \\
Method & $0.60$ & $0.95$ & $0.99$ & $0.60$ & $0.95$ & $0.99$ & AUC\\
\hline
\multicolumn{8}{l}{\emph{Four-marker (frozen \amk{4MP} vs.\ iPEB on the same four markers)}}\\
PEB (frozen 4MP)      & 0.949 & 0.434 & 0.182 & 2.60 & 2.30 & 1.09 & 0.875\\
iPEB (sensitivity)    & 0.939 & 0.485 & 0.242 & 2.59 & 2.01 & 1.18 & 0.877\\
iPEB (lead time)      & 0.939 & 0.505 & 0.242 & 2.49 & 2.01 & 1.21 & 0.875\\
iPEB (combined)       & 0.939 & 0.505 & 0.242 & 2.49 & 2.01 & 1.21 & 0.874\\
\hline
\multicolumn{8}{l}{\emph{Six-marker (frozen logistic vs.\ iPEB on the same six markers)}}\\
PEB (frozen logistic) & 0.960 & 0.414 & 0.162 & 2.59 & 2.30 & 1.06 & 0.869\\
iPEB (sensitivity)    & 0.939 & 0.485 & 0.242 & 2.49 & 2.01 & 1.18 & 0.876\\
iPEB (lead time)      & 0.939 & 0.495 & 0.242 & 2.49 & 1.93 & 1.21 & 0.876\\
iPEB (combined)       & 0.939 & 0.455 & 0.202 & 2.59 & 1.50 & 0.71 & 0.870\\
\hline
\end{tabular}
\end{table}

\begin{table}[htbp]\centering\footnotesize
\caption{\rev{PLCO lung cohort, \emph{training} centers (in-sample companion to \hyperref[tab:supp-lung-full]{Web Table~3}; six centers, 225 case / 1{,}177 control subjects): four-marker and six-marker comparisons across all three iPEB objectives, same layout and calibration as the test table.}}
\label{tab:supp-lung-full-train}
\begin{tabular}{l ccc ccc c}
\hline
 & \multicolumn{3}{c}{Sensitivity (by spec.)} & \multicolumn{3}{c}{Lead time, yr (by spec.)} & \\
Method & $0.60$ & $0.95$ & $0.99$ & $0.60$ & $0.95$ & $0.99$ & AUC\\
\hline
\multicolumn{8}{l}{\emph{Four-marker (frozen \amk{4MP} vs.\ iPEB on the same four markers)}}\\
PEB (frozen 4MP)      & 0.920 & 0.409 & 0.116 & 2.63 & 1.65 & 0.55 & 0.858\\
iPEB (sensitivity)    & 0.938 & 0.431 & 0.160 & 2.49 & 1.71 & 0.70 & 0.870\\
iPEB (lead time)      & 0.942 & 0.427 & 0.156 & 2.54 & 1.85 & 0.57 & 0.870\\
iPEB (combined)       & 0.942 & 0.453 & 0.156 & 2.54 & 1.72 & 0.57 & 0.869\\
\hline
\multicolumn{8}{l}{\emph{Six-marker (frozen logistic vs.\ iPEB on the same six markers)}}\\
PEB (frozen logistic) & 0.920 & 0.382 & 0.116 & 2.66 & 1.84 & 0.91 & 0.854\\
iPEB (sensitivity)    & 0.938 & 0.431 & 0.173 & 2.49 & 1.71 & 0.83 & 0.870\\
iPEB (lead time)      & 0.942 & 0.431 & 0.160 & 2.52 & 1.73 & 0.70 & 0.869\\
iPEB (combined)       & 0.942 & 0.396 & 0.160 & 2.52 & 1.44 & 0.51 & 0.866\\
\hline
\end{tabular}
\end{table}

\begin{table}[htbp]\centering\footnotesize
\caption{\rev{PLCO lung cohort, selected-four comparison across all three iPEB objectives, on both the test and training centers (companion to \hyperref[tab:supp-lung-full]{Web Tables~3--4}). Each method reduces the six markers to four on the training centers---PEB by logistic $z$-statistic ranking (objective-independent, hence a single row), iPEB by objective-driven backward elimination stopped at four. Sensitivity (Sens) and median lead time in years (Lead) at three specificities, with AUC. The logistic ranking recovered exactly the four published 4MP markers (CA125, CEA, CYFRA~21-1, pro-SFTPB), on which the PEB side refits a logistic on the training centers (log-normalized markers): intercept $-13.14$, CA125 $0.501$, CEA $0.935$, CYFRA~21-1 $0.350$, pro-SFTPB $1.535$. iPEB selected CEA, pro-SFTPB, osteopontin and HE4 under the sensitivity objective, CEA, CYFRA~21-1, osteopontin and HE4 under the lead-time objective, and CA125, CYFRA~21-1, osteopontin and HE4 under the combined objective. The three objectives therefore give three distinct selected panels and three distinct rows.}}
\label{tab:supp-lung-sel}
\begin{tabular}{l ccc ccc c}
\hline
 & \multicolumn{3}{c}{Sensitivity (by spec.)} & \multicolumn{3}{c}{Lead time, yr (by spec.)} & \\
Method & $0.60$ & $0.95$ & $0.99$ & $0.60$ & $0.95$ & $0.99$ & AUC\\
\hline
\multicolumn{8}{l}{\emph{Test centers}}\\
PEB (logistic-selected 4)  & 0.960 & 0.455 & 0.172 & 2.59 & 2.11 & 1.13 & 0.871\\
iPEB (sensitivity)         & 0.919 & 0.515 & 0.192 & 2.64 & 1.90 & 1.07 & 0.873\\
iPEB (lead time)           & 0.929 & 0.535 & 0.242 & 2.55 & 1.90 & 1.13 & 0.887\\
iPEB (combined)            & 0.889 & 0.364 & 0.162 & 2.28 & 1.49 & 1.16 & 0.832\\
\hline
\multicolumn{8}{l}{\emph{Training centers (in-sample)}}\\
PEB (logistic-selected 4)  & 0.924 & 0.396 & 0.107 & 2.68 & 1.82 & 0.84 & 0.853\\
iPEB (sensitivity)         & 0.893 & 0.373 & 0.138 & 2.59 & 1.54 & 0.75 & 0.843\\
iPEB (lead time)           & 0.889 & 0.356 & 0.173 & 2.62 & 1.57 & 0.83 & 0.832\\
iPEB (combined)            & 0.880 & 0.333 & 0.120 & 2.34 & 1.45 & 1.18 & 0.827\\
\hline
\end{tabular}
\end{table}

\begin{table}[htbp]\centering\footnotesize
\caption{\rev{PLCO lung cohort: subject-level cross-classification of the case subjects at each operating specificity, on the held-out test centers ($n=99$ cases) and the training centers ($n=225$ cases), counting cases detected by both methods, by iPEB only, by the frozen/selected panel only, or by neither, for all three comparisons. iPEB uses the sensitivity objective and the threshold is calibrated on training controls; this unweighted count accompanies the tabulated sensitivities and shows the net movement of individual patients.}}
\label{tab:supp-lung-paired}
\begin{tabular}{lll cccc}
\hline
Comparison & Set & Spec & Both & iPEB only & Panel only & Neither\\
\hline
\multirow{6}{*}{Four-marker}
 & Test  & 0.60 & 91  & 2  & 3  & 3\\
 &       & 0.95 & 39  & 9  & 4  & 47\\
 &       & 0.99 & 16  & 8  & 2  & 73\\
\cline{2-7}
 & Train & 0.60 & 202 & 9  & 5  & 9\\
 &       & 0.95 & 75  & 22 & 17 & 111\\
 &       & 0.99 & 24  & 12 & 2  & 187\\
\hline\hline
\multirow{6}{*}{Six-marker}
 & Test  & 0.60 & 92  & 1  & 3  & 3\\
 &       & 0.95 & 37  & 11 & 4  & 47\\
 &       & 0.99 & 13  & 11 & 3  & 72\\
\cline{2-7}
 & Train & 0.60 & 201 & 10 & 6  & 8\\
 &       & 0.95 & 69  & 28 & 17 & 111\\
 &       & 0.99 & 22  & 17 & 4  & 182\\
\hline\hline
\multirow{6}{*}{Selected-four}
 & Test  & 0.60 & 91  & 0  & 4  & 4\\
 &       & 0.95 & 37  & 14 & 8  & 40\\
 &       & 0.99 & 12  & 7  & 5  & 75\\
\cline{2-7}
 & Train & 0.60 & 193 & 8  & 15 & 9\\
 &       & 0.95 & 63  & 21 & 26 & 115\\
 &       & 0.99 & 16  & 15 & 8  & 186\\
\hline
\end{tabular}
\end{table}

\bibliographystyle{apalike}
\bibliography{bibliography}